\documentclass[runningheads]{llncs}
\usepackage[a4paper, margin=1.2in]{geometry}

\usepackage[T1]{fontenc}
\usepackage{lmodern}
\usepackage{amsmath,amssymb,mathtools,centernot,nicefrac}
\usepackage{graphicx}
\usepackage{booktabs,multirow}
\usepackage{array}
\usepackage{enumitem}
\usepackage{xspace,xcolor}
\usepackage[
    colorlinks=true,
    linkcolor=blue,
    citecolor=blue,
    urlcolor=blue,
    pdfborder={0 0 0}
]{hyperref}
\usepackage[capitalise,noabbrev]{cleveref}

\newcommand{\pp}{\mathsf{pp}}

\newcommand{\PPT}{\ensuremath{\mathsf{PPT}}\xspace}

\newcommand{\N}{\ensuremath{\mathbb{N}}}
\newcommand{\R}{\ensuremath{\mathbb{R}}}
\newcommand{\Z}{\ensuremath{\mathbb{Z}}}

\newcommand{\G}{\ensuremath{\mathbb{G}}}

\newcommand{\GG}{\ensuremath{\mathcal{G}}}
\newcommand{\edges}{\ensuremath{\mathcal{E}}}
\newcommand{\vertices}{\ensuremath{\mathcal{V}}}

\newcommand{\MMS}{\ensuremath{\Sigma}}
\newcommand{\ZKP}{\ensuremath{\Pi}}

\newcommand{\algorithm}[1]{\mathsf{#1}}

\newcommand{\KGen}{\algorithm{KGen}}
\newcommand{\Issue}{\algorithm{Issue}}
\newcommand{\Sign}{\algorithm{Sign}}
\newcommand{\Present}{\algorithm{Present}}
\newcommand{\Verify}{\algorithm{Verify}}
\newcommand{\ParGen}{\algorithm{ParGen}}

\newcommand{\VerifyCert}{\algorithm{CertVf}}
\newcommand{\ValidPath}{\algorithm{ValidPath}}
\newcommand{\Sim}{\algorithm{Sim}}
\newcommand{\Ext}{\algorithm{Ext}}

\newcommand{\e}{\mathsf{e}} %Pairing

\newcommand{\variable}[1]{\mathsf{#1}}
\newcommand{\id}{\variable{id}}
\newcommand{\sk}{\variable{sk}}
\newcommand{\pk}{\variable{pk}}
\newcommand{\isk}{\variable{isk}}
\newcommand{\ipk}{\variable{ipk}}

\newcommand{\ctx}{\variable{ctx}}
\newcommand{\state}{\variable{st}}
\newcommand{\ctr}{\variable{cid}}
\newcommand{\cert}{\variable{cert}}

\newcommand{\policy}{\phi}
\newcommand{\sig}{\sigma}

\newcommand{\rnd}{t}
\newcommand{\attr}{a}
\newcommand{\source}{S}
\newcommand{\dest}{D}
\newcommand{\node}{\texttt{n}}
\newcommand{\edge}{\texttt{e}}
\newcommand{\selfloop}{\texttt{loop}}

\newcommand{\secpar}{\ensuremath{\lambda}}
\newcommand{\rexec}{\ensuremath{\xleftarrow{\scriptscriptstyle \$}}}
\newcommand{\exec}{\ensuremath{\leftarrow}}
\newcommand{\negl}{\ensuremath{\nu}}
\newcommand{\bin}{\ensuremath{\{0,1\}}}

\newcommand{\NIZK}{\ensuremath{\mathsf{NIZK}}}

\newcommand{\adv}{\mathcal{A}}

\newcommand{\oracle}{\mathcal{O}}
\newcommand{\Ocert}{\oracle_\text{cert}}
\newcommand{\Ocorr}{\oracle_\text{corr}}
\newcommand{\Opres}{\oracle_\text{pres}}

\newcommand{\setQ}{\mathcal{Q}}
\newcommand{\Qcert}{\setQ_\text{cert}}
\newcommand{\Qpres}{\setQ_\text{pres}}
\newcommand{\Qcorr}{\setQ_\text{corr}}

\begin{document}

\pagestyle{plain}

\title{Certifying Hidden Paths: \\Scalable Topology Assurance for QKD Networks}

\author{
Alessandro Colombo\inst{1}\orcidID{0009-0009-4574-7101}
\and
Margherita Cozzolino\inst{2}\thanks{This work was partially carried out while the author was with AIT Austrian Institute of Technology.}
\and \\
Stephan Krenn\inst{1}\orcidID{0000-0003-2835-9093}
\and
Thomas Lorünser\inst{1,3}\orcidID{0000-0002-1829-4882}
}

\authorrunning{A. Colombo et al.}

\institute{
AIT Austrian Institute of Technology, Vienna, Austria\\
\email{\{alessandro.colombo,stephan.krenn,thomas.loruenser\}@ait.ac.at}
\and
Independent Researcher, Trento, Italy
\and
Digital Factory Vorarlberg GmbH, Dornbirn, Austria
}

\maketitle

\begin{abstract}
Large-scale Quantum Key Distribution (QKD) networks rely on trusted repeaters, making the security properties of the selected communication path an essential part of end-to-end assurance.
At the same time, network operators may be unwilling to disclose their internal topology.
We present a topology-certification mechanism that lets a provider prove in zero knowledge that policy-compliant routes between two endpoints exist, without disclosing the routes in an individual presentation.
Our main idea is to certify nodes and edges independently using multi-message signatures, rather than signing the complete graph as one object.
For a route chosen in advance, proof size and cryptographic proving and verification work depend on its positions and attributes, independently of the overall network size, whereas route discovery and certification have separate graph-dependent costs.
The construction hides the actual route length up to a public bound and sketches extensions to node-disjoint routes and monotonic additions within a graph epoch.
We give a formal security model and conditional proofs of unforgeability and graph hiding for a generic construction.
For a BBS-based construction, presentation-size estimates remain below $300$\,KiB at $\ell=m=n=50$, and measured generation and verification times remain below $400$\,ms at $\ell=64$ and $m=n=8$.

\keywords{Quantum key distribution \and topology privacy \and path certification \and zero-knowledge proofs \and trusted repeater networks}
\end{abstract}

\section{Introduction}\label{s:introduction}

Quantum Key Distribution (QKD) establishes shared keys over optical links with information-theoretic security under its protocol assumptions. It complements post-quantum cryptography in addressing long-term threats, but optical loss limits the reach of individual links. Although quantum repeaters may eventually lift this restriction, current large-scale networks rely on trusted repeaters to connect distant users. Neighbors generate link keys, and intermediate devices relay key material between them. A repeater on the selected route can reconstruct the end-to-end key, so link-level QKD alone does not determine connection security.

The physical protection, hardware integrity, and operational procedures of every traversed repeater therefore matter. Regular audits can certify devices and deployment conditions, but different sites and operators may offer different levels of assurance. A customer's relevant question is whether the components of an available route \emph{together} meet its requirements, rather than whether isolated devices carry certificates.

A customer may require approved vendors, minimum certification levels at each node, sufficient certified capacity on genuine links, and a bounded number of hops.
For governmental deployments, such requirements may be checked by an authority that is allowed to inspect the complete network architecture.
For private QKD customers, however, this is generally not an option:
they cannot be expected to obtain access to the provider's internal topology, which may contain commercially sensitive and security-critical information about sites, interconnections, and central bottlenecks~\cite{DBLP:conf/uss/MeierTLVV18}.
An independent auditor may still certify the underlying infrastructure, but the provider needs a way to later answer customer requests without disclosing the audited graph.
The resulting assurance goal can therefore be phrased as a single statement:
\emph{Some certified path between the chosen endpoints satisfies the requested properties specified in a policy.}

A first solution to this problem was recently proposed by \cite{qcnc2026}.
Their construction is based on graph signatures proposed by \cite{DBLP:journals/iacr/TanSG20} to enable topology-hiding path certification.
In their setting, an independent auditor certifies the complete network graph, including its topology as well as all node and edge attributes.
Based on this certification, a network provider can subsequently generate a zero-knowledge proof demonstrating that the certified network contains a path satisfying a customer's requested policy, while revealing neither the network topology nor the selected path itself.

While this work demonstrates the feasibility of topology-hiding certification for QKD networks, several practical limitations remain.
\begin{enumerate}[label=(\roman*)]
  \item 
    First, the proof size as well as prover and verifier costs grow linearly with the size of the certified network graph, making the approach increasingly costly for larger deployments.
  \item 
    Second, the paper mainly focuses on path availability, rather than complex policies on the path, such that the precise set of supported policies remains unclear, a challenge that was subsequently addressed in \cite{thesis-margherita}, yet without giving formal guarantees.
  \item 
    Third, path constraints and attribute-based constraints cannot be naturally combined into a unified policy framework, especially when aiming for disjoint multi-path communications.
  \item 
    Finally, while the construction in \cite{qcnc2026} focuses on inter-network communication, it remains unclear how efficient path-length hiding can be achieved for intra-network communication, due to the employed padding mechanism.
    However, concealing the path length is particularly important in this setting, as the number of traversed repeaters may itself reveal sensitive information about the underlying topology, including the existence of direct links and other structural characteristics.
    This issue was partially addressed by \cite{thesis-margherita}, yet leaving it open how to combine attributes with padding approaches.
\end{enumerate}

These limitations motivate the need for more expressive and scalable topology certification mechanisms that preserve the confidentiality of the provider's infrastructure while enabling customers to obtain strong, independently verifiable guarantees about the existence of policy-compliant communication paths.

\subsection{Our Contributions}

We present a topology-certification mechanism for QKD networks that allows a provider to prove in zero knowledge that one or several suitable paths between two endpoints exist, while hiding the network topology and the selected paths.
The construction is tailored to the size and requirements of large-scale trusted-repeater QKD networks and supports policies over both nodes and edges.

Our main technical idea is to certify nodes and edges independently instead of authenticating the complete network as a single graph object.
Compared with \cite{qcnc2026} and \cite{thesis-margherita}, this makes presentation size and computation independent of the overall network size and dependent only on the presented path and its attributes.
It also allows certified networks to grow efficiently, since new nodes and edges can be added without re-certifying the complete graph.
Importantly, instead of intuitive security goals in previous constructions, we provide a formal security framework and detailed proofs to ensure security of our construction.

The construction further supports hiding the actual path length through dedicated padding loops.
In contrast to~\cite{qcnc2026}, where an upper bound on the padding length has to be fixed already at certification time and the presentation cost grows with this bound, our construction does not require such a predefined maximum.
Supporting padding adds only a fixed quantifiable certification overhead, while the proof cost depends solely on the padding length chosen for the individual presentation.
This length can therefore be selected freely for each proof without re-certifying the graph.
Addressing a further open challenge from previous work, we show how this padding mechanism can also be combined with our extension for proving several hidden node-disjoint paths.

Finally, we provide a concrete instantiation based on BBS multi-message signatures and Schnorr-based zero-knowledge proofs.
BBS is particularly relevant for deployment, as it is currently considered in various standardization efforts.
Our evaluation shows that presentation sizes remain below $300$\,KiB for parameter ranges exceeding realistic QKD deployments, while presentation generation and verification remain below approximately $400$\,ms even for paths of lengths exceeding those to be expected in typical QKD applications.

\subsection{Technical Overview}\label{s:technical_overview}

On an informal level our scheme works as follows. 
Before a customer requests a connection, an issuer audits a particular network of a provider which is modeled as a directed network graph.
It assigns a graph identifier and a fresh random binding tag, then signs each node's identifier and attributes and each edge's endpoint identifiers and attributes separately.
Every signature includes the same graph identifier and tag.
The common tag ensures that a proof cannot combine components from different certification runs.
The issuer requires distinct node identifiers and adds a signed self-loop with distinguished attributes at each node.
These loops will later serve for padding in order to hide the path length.

To generate a request, the customer publishes source and destination identifiers, predicates for acceptable nodes and genuine links, a bound $\ell$ on the number of node positions, and potentially a request context $\ctx$. The provider privately finds a compliant route passing through at most $\ell$ nodes. If it is shorter, the provider repeats the destination node using its certified self-loop until the sequence has exactly $\ell$ node positions. Thus the proof always has the public shape of $\ell$ nodes and $\ell-1$ edges, regardless of the route's actual length.

The provider then produces one non-interactive zero-knowledge proof.
It shows knowledge of valid issuer signatures for all components of a valid path, bound to the same graph and binding tag.
For every edge position, the signed endpoint identifiers must match the identifiers of the adjacent signed nodes, and every node must satisfy the required node policy.
Each edge must either satisfy the link policy or have equal endpoint identifiers, in which case it represents a padding loop.
The proof hides which of these two conditions is satisfied at each edge position.

Finally, the customer verifies this proof against the issuer's public keys and its request. Acceptance establishes the existence of a matching route in the certified graph while hiding its internal components and actual length.
Once a route is selected, the proof core scales with the padded route and its attributes, whereas certification storage and route discovery depend on the full graph.

\subsection{Related Work}\label{s:related_work}

Most closely related to our work are recent approaches for topology-hiding assurance in QKD networks.
\cite{qcnc2026} and \cite{thesis-margherita} introduce topology-hiding connectivity assurance for QKD inter-networking, enabling network providers to prove connectivity between endpoints without disclosing their internal network topologies.
Their construction considers multi-graphs and hidden endpoints and additionally supports statements about multiple disjoint paths.
However, in their constructions, the computational and communication complexity scales with the size of the overall QKD network, and many extensions, such as support of multi-path transmissions, are only sketched.
\cite{acns2026} subsequently study topology-hiding path validation within large-scale QKD networks and provide a formal framework in which a network operator proves that the path used for a transmission satisfied certain pre-established conditions.
Their framework provides a-posteriori assurance, whereas ours provides a-priori assurance.

The construction of \cite{qcnc2026} builds on a broader line of work on graph signatures and zero-knowledge proofs over certified graphs.
\cite{DBLP:conf/ccs/Gross14} introduces efficient certification and zero-knowledge proofs for infrastructure topology graphs, allowing properties such as connectivity and isolation to be proven without revealing the certified graph.
\cite{DBLP:conf/fc/Gross15} develops this approach further by considering signatures and efficient proofs over committed graphs and more general NP statements.
\cite{DBLP:journals/corr/abs-2005-12447} provide an implementation and evaluation of the strong-RSA-based graph-signature construction.
An alternative graph-signature scheme based on the $q$-SDH assumption was proposed in \cite{DBLP:journals/iacr/TanSG20} and later extended towards relational credentials by \cite{DBLP:journals/iacr/TanSG23}.
\cite{DBLP:conf/ccs/NakanishiYMP20} improves proofs for connectivity and isolation using bilinear-map accumulators, in particular obtaining proof sizes and verification costs that are independent of the total number of vertices and edges.
This work was extended by \cite{DBLP:journals/ieiceta/YoshiokaNK25} to labeled directed graphs.
These works demonstrate that efficient privacy-preserving proofs over certified graphs are possible under a variety of cryptographic assumptions and representations.
Our approach differs in that it avoids treating the complete network as a single certified graph object and instead builds topology assurance from independently authenticated graph components.

The confidentiality of network topology has also been considered independently of QKD.
\cite{DBLP:journals/tcs/HinkelmannJ07} studies communication in networks whose topology itself constitutes sensitive information and shows that communication may inherently reveal structural information about the underlying network.
This observation is particularly relevant for infrastructure networks, where connectivity, path lengths, and central nodes may expose operational or security-sensitive information.
More recently, \cite{DBLP:journals/tdsc/LiSLSMS25} proposed StealthPath for privacy-preserving path validation in path-aware networks, hiding both the forwarding path and its length while allowing validation in the data plane.
This setting differs from ours in that the objective is to validate that packets traverse an intended path, whereas we prove the existence and properties of a path in a certified but hidden network.

The security of trusted-repeater QKD networks has been studied from several complementary perspectives.
\cite{DBLP:journals/jcs/SalvailPDALL10} analyzes the security of trusted-repeater QKD networks in the presence of compromised nodes and considers mechanisms for maintaining confidentiality and authenticity despite partially untrusted infrastructure.
A common approach to reducing the trust placed in individual repeaters is to distribute key material across multiple paths.
\cite{DBLP:conf/IEEEares/ValbusaLSL25} surveys such multi-path approaches and discusses how cryptographic and coding techniques can mitigate the single-point-of-failure problem inherent in trusted-repeater networks.
Our extension to node-disjoint path proofs is complementary to these techniques, since it allows a verifier to establish the required disjointness without learning the identities of the intermediate nodes.
Finally, \cite{qcnc2025} studies secure key forwarding in large-scale QKD networks and provides confidentiality and authenticity guarantees for transporting key material along a selected path.
While that work addresses the secure forwarding of keys once a path has been determined, the present work addresses the preceding assurance problem of proving that a suitable path with the required properties exists.

\subsection{Outline}

The remainder of the paper is structured as follows.
Section~\ref{s:preliminaries} introduces the preliminaries, and Section~\ref{s:model} defines the syntax and security model.
Section~\ref{s:abcs} presents the construction, security analysis, and extensions.
Section~\ref{s:evaluation} evaluates presentation size and runtime, and Section~\ref{s:conclusion} presents a short summary and outlook.
The supplementary material contains full proofs and additional technical details.

\section{Preliminaries}\label{s:preliminaries}

\subsection{Notation}\label{s:notation}
The main security parameter is $\secpar\in\N$. Unless stated otherwise, algorithms receive $1^\secpar$ implicitly and run in probabilistic polynomial time (\PPT) in $\secpar$ and the length of their explicit inputs. We write $a\rexec A(x)$ for the output of a randomized algorithm and $a\exec A(x)$ for deterministic execution. When its coins matter, we write $a\rexec A(x;r)$. The notation $a\rexec S$ means uniform sampling from a finite set $S$. We use $[k]=\{1,\ldots,k\}$ and $\Z_p$ for the integers modulo $p$. A function $\negl:\N\to\R_{\geq0}$ is negligible if, for every constant $c>0$, it is eventually smaller than $\secpar^{-c}$. Tuples, public statements, and challenges have fixed, unambiguous, efficiently decodable encodings.

The generic construction does not require a particular algebraic group. Its concrete BBS proof core uses groups $\G_1,\G_2,\G_T$ of prime order $p$ and a non-degenerate bilinear pairing $\e:\G_1\times\G_2\to\G_T$. We reserve $\GG=(\vertices,\edges)$ for a network graph, distinguishing it from the groups $\G_i$. The certified components of a selected route and their signatures belong to the prover's witness, whereas endpoints, policy predicates, a padded length bound, and a request context are public inputs to a presentation.

\subsection{Cryptographic Building Blocks}\label{s:building_blocks}

Our construction uses two cryptographic building blocks. Multi-message signatures authenticate the certified nodes and edges of a network, while non-interactive proofs let a provider show that these components form a valid route without disclosing it. We introduce the syntax of both primitives and the security notions needed for our analysis: EUF-CMA security for signatures, and zero knowledge and fresh-statement weak simulation extractability for proofs. Their full definitions, correctness conditions, and security games appear in the supplementary material (\cref{s:supplement_proof_games}).

\begin{Definition}[Multi-message signatures]\label{def:multi-message-signatures}
A scheme signs vectors $\vec m=(m_1,\ldots,m_L)\in\mathcal M^L$ of supported length $L$.
Parameter generation $\ParGen(1^\secpar,L)$ returns $\pp$, and key generation $\KGen(\pp)$ returns $(\sk,\pk)$.
Signing $\Sign(\sk,\vec m)$ returns $\sig$, and verification $\Verify(\pk,\vec m,\sig)$ returns a bit.
Here $\pp$ is implicit in $\Sign$ and $\Verify$. Correctness is perfect: every honestly generated signature on a supported vector verifies.
\end{Definition}

\begin{Definition}[EUF-CMA security]\label{def:euf-cma}
For each supported $L$, generate parameters and keys honestly and give a \PPT adversary $\adv$ $(\pp,\pk)$ and adaptive access to $\Sign(\sk,\cdot)$. Let $Q$ be its queried vectors and $(\vec m^*,\sig^*)$ its output. Following~\cite{DBLP:journals/siamcomp/GoldwasserMR88}, the scheme is EUF-CMA secure if, for every $\adv$, there exists a negligible function $\negl$ such that
\[
  \Pr\bigl[\Verify(\pk,\vec m^*,\sig^*)=1
  \ \land\ \vec m^*\notin Q\bigr]\leq\negl(\secpar)\,.
\]
A new signature on an already queried vector is not a forgery.
\end{Definition}

\begin{samepage}
For our application, the issuer signs nodes and directed edges under separate keys. Their signed vectors contain a graph identifier and a fresh graph-wide binding tag together with identifiers and attributes. EUF-CMA rules out a valid signature on a component vector that the issuer did not certify, even after the adversary obtains certificates for graphs of its choice. The tag and endpoint identifiers are signed messages, rather than assumptions supplied only by the prover.

\medskip

For a relation $\mathcal R$, let $\mathcal L_{\mathcal R}=\{x\mid\exists w:\mathcal R(x,w)=1\}$ denote its language.

\begin{Definition}[Non-interactive zero-knowledge proofs]\label{def:nizk}
Let $\mathcal R(x,w)$ be an efficiently decidable relation with public statement $x$ and private witness $w$. A non-interactive proof system comprises $(\ParGen,\NIZK,\Verify)$, with $\pp\rexec\ParGen(1^\secpar)$, $\pi\rexec\NIZK(\pp,x,w)$ for $\mathcal R(x,w)=1$, and acceptance decided by $\Verify(\pp,x,\pi)$. Perfect completeness requires every honest proof for a valid pair $(x,w)$ to verify. In the path relation below, $x$ specifies the issuer key, graph identifier, endpoints, policies, public bound and context, whereas $w$ contains certified nodes and edges, their hidden attributes, signatures, and common binding tag.
\end{Definition}
\end{samepage}

We write a proof of knowledge using the notation introduced in \cite{DBLP:conf/crypto/CamenischS97}:
\[
  \pi\rexec\NIZK[(w):\mathcal R(x,w)=1](\ctx)\,.
\]
Here $w$ denotes the secret witness, $x$ is public, and the context $\ctx$ is bound to the statement and proof, which can be used, e.g., to ensure freshness of the proof.

\begin{Definition}[Zero knowledge~\cite{DBLP:conf/stoc/BlumFM88,DBLP:conf/cans/BobolzEKRS21}]\label{def:zero-knowledge}
There are \PPT simulators $(\Sim_1,\Sim_2)$ such that $\Sim_1(1^\secpar)$ generates $(\pp,\tau_{\mathsf{sim}})$ and $\Sim_2(\tau_{\mathsf{sim}},x)$ generates a proof without $w$. No \PPT distinguisher that chooses a valid pair $(x,w)$ after seeing $\pp$ distinguishes real parameters and proofs from simulated ones with more than negligible advantage.
\end{Definition}

In our context, for a fixed public request, the proof shall not disclose which certified route supplied its witness. The public endpoint and policy values remain visible.

\medskip

Weak simulation extractability~\cite{DBLP:conf/indocrypt/FaustKMV12} intuitively means that, even after an adversary has seen arbitrarily many simulated proofs, any new accepting proof it produces must still correspond to a valid witness that an extractor can recover. 
The extractor is allowed to rewind the adversary, so the notion is weaker than straight-line simulation extractability but still strong enough for security reductions based on Fiat--Shamir proofs.
In the unforgeability proof of our generic construction, the extracted witness supplies the individual signatures and a complete path. 
Either those signatures authenticate an actual route in a certified graph or one of them authenticates a previously unsigned vector. This reduction uses extractability after simulated presentation-oracle replies, which ordinary soundness or zero knowledge alone does not provide. 

However, it turns out that this property is slightly too strong in our setting, as we only consider forgeries for fresh statements, as we will see in the following sections.
We thus define a weaker notion of this property as follows:

\begin{Definition}[Fresh-statement weak simulation extractability]
\label{def:wse}
A non-interactive proof system in the random-oracle model is \emph{fresh-statement weakly simulation extractable} if, for every \PPT adversary with adaptive access to the proof simulator, there exists an efficient extractor that may rewind the adversary and has access to its random tape and the simulator and random-oracle transcripts.
Whenever the adversary produces such a fresh accepting proof $\pi^*$ with more than negligible probability, the extractor outputs a corresponding valid witness $w^*$ with non-negligible probability polynomially related to the adversary's success probability.
\end{Definition}
%\begin{Definition}[Simulation-sound extractability~\cite{DBLP:conf/asiacrypt/Groth06,DBLP:conf/cans/BobolzEKRS21}]\label{def:sse}
%There are \PPT algorithms $(\Sim_1,\allowbreak\Sim_2,\Ext)$ with joint simulated setup $(\pp,\tau_{\mathsf{sim}},\tau_{\mathsf{ext}})\rexec\Sim_1(1^\secpar)$. For any \PPT adversary querying $\Sim_2(\tau_{\mathsf{sim}},\cdot)$ on statements in $Q$, an accepting proof $\pi^*$ for $x^*\notin Q$ yields $w^*\rexec\Ext(\tau_{\mathsf{ext}},x^*,\pi^*)$ satisfying $\mathcal R(x^*,w^*)=1$, except with negligible probability.
%\end{Definition}

\section{Syntax and Security Model}\label{s:model}
We formalize two guarantees for a provider presenting evidence about an audited QKD network: an accepting token must correspond to a certified compliant route, while hiding the selected route and internal topology beyond the public request. Certification concerns a network snapshot, and a later presentation answers a customer's specific request.

We model the network as a directed graph $\GG=(\vertices,\edges)$, representing an undirected link by two directed edges with the same attributes. Every node $\node$ has an identifier $\id^\node$. Endpoint identifiers may be public while internal identifiers and node counts remain private. The map $\node\mapsto\id^\node$ is injective within each certified graph, though not necessarily across graphs. 

A node carries $n$ attributes and an edge carries $m$ attributes, authenticated and tested against public predicates $\policy_\node$ and $\policy_\edge$. Examples include an approved repeater vendor and certified link capacity. The graph identifier $\id_\GG$ designates the audited graph or epoch. 

\subsection{Syntax}\label{s:syntax}
A scheme has six probabilistic polynomial-time algorithms:
$\ParGen$, $\KGen$, $\Issue$, $\VerifyCert$, $\Present$, and $\Verify$.
Write $\mathsf A=(\mathsf A_\node,\mathsf A_\edge)$ for the node and edge attribute maps, and $g=\id_\GG$. 

\begin{description}[nosep]
\item[Parameter generation.] $\pp\rexec\ParGen(1^\secpar,n,m)$ sets public cryptographic and attribute parameters.
\item[Key generation.] $(\isk,\ipk)\rexec\KGen(\pp)$ returns the issuer's secret and public keys.
\item[Certification.] $\cert\rexec\Issue(\GG,\id_\GG,\{(\attr_i^\node)_{i=1}^n\}_{\node\in\vertices},\{(\attr_i^\edge)_{i=1}^m\}_{\edge\in\edges},\isk)$ certifies a graph, graph identifier, node and edge attributes $\{(\attr_i^\node)_{i=1}^n\}_{\node\in\vertices}$ and $\{(\attr_i^\edge)_{i=1}^m\}_{\edge\in\edges}$, and may reject malformed inputs. The provider retains $\cert$.
\item[Certificate validation.] Given $\cert,\GG,\id_\GG,\{(\attr_i^\node)_{i=1}^n\}_{\node\in\vertices},\{(\attr_i^\edge)_{i=1}^m\}_{\edge\in\edges}$ and $\ipk$, algorithm $\VerifyCert$ checks whether the provided certificate is valid for the given graph.
\item[Presentation.] Given $\cert,\ipk,\GG,\id_\GG,\{(\attr_i^\node)_{i=1}^n\}_{\node\in\vertices},\{(\attr_i^\edge)_{i=1}^m\}_{\edge\in\edges},\policy_\node,\allowbreak\policy_\edge,\allowbreak\id^{\node_\source},\allowbreak\id^{\node_\dest},\allowbreak\ell,\allowbreak\ctx$ as inputs, algorithm $\Present$ returns a token $\pi$ for a feasible request.
\item[Verification.] $\Verify(\pi,\ipk,\id_\GG,\policy_\node,\policy_\edge,\id^{\node_\source},\id^{\node_\dest},\ell,\ctx)$ checks the public request without seeing the graph, certificate, or route.
\end{description}

Correctness requires an honestly issued certificate to pass $\VerifyCert$ and an honest token for a feasible request to pass $\Verify$, except with negligible probability. The public context $\ctx$ binds a token to a request. Applications can include a session identifier, but the scheme does not itself supply freshness.

Here $\ell\geq 1$ counts \emph{node positions}, including endpoints, in a padded witness: a real route uses at most $\ell-1$ genuine edges. The predicate $\ValidPath$ checks that route in the input graph before padding. The construction fills remaining positions with issuer-certified destination self-loops. Padding loops are exempt from the genuine-link predicate, while the repeated destination still satisfies the node predicate. Thus the public $\ell$ need not reveal the actual route length.

\subsection{Security Model}\label{s:security_model}
On a high level, we require two security properties: firstly, \emph{unforgeability} guarantees that every valid presentation was derived from a certificate that was previously issued for the given graph identifier. 
Secondly, \emph{graph hiding} requires that by merely looking at the presentation token, it should not be possible to infer any additional information about the underlying graph.
For notational convenience, the security definitions will rely on an auxiliary algorithm $\ValidPath$, which takes as input a graph, its associated node and edge attributes, node and edge policies, source and destination node identifiers, and a maximum path length $\ell$. 
It outputs $1$ if there exists a path from the source to the destination of length at most $\ell$ such that every node and edge on the path satisfies the corresponding policy, and outputs $0$ otherwise.

\subsubsection{Unforgeability}
We require that no adversary can generate a valid presentation $\pi^*$ without having
received a certificate for the related policies. 
More formally, akin to unforgeability of attribute-based credentials, we allow the adversary to obtain arbitrary credentials on graphs of their choice. 
Furthermore, the adversary may request presentation tokens from uncorrupted certificates. 
Eventually, the adversary outputs a presentation token $\pi^*$, together with a graph identifier $\id_\GG$, policies $\policy_\node^*$ and $\policy_\edge^*$ , source and destination nodes $\id^{\node_S*}$ and $\id^{\node_D*}$ as well as target length $\ell^*$ and context $\ctx^*$. 
The adversary wins if $\pi^*$ is valid, if for the given presentation requirements no other presentation token has been queried, and if no corrupted certificate satisfies the required policies and path requirements.

\begin{Definition}[Unforgeability]\label{def:unforgeability}
A scheme is unforgeable if, for every probabilistic polynomial-time adversary $\adv$ and supported attribute counts $n,m$, there is a negligible function $\negl$ such that
\[
\Pr\bigl[\mathrm{Exp}_{\mathrm{unforgeability}}^{\adv}(\secpar,n,m)=1\bigr]
\leq\negl(\secpar).
\]
The oracles and precise winning conditions appear in \cref{fig:unforgeability}.
\end{Definition}

\begin{figure}[th!]
	\centering\renewcommand{\arraystretch}{1.2}
	\begin{tabular}{|l|}
		\hline
		\rule{0pt}{2.5ex}\textbf{Experiment} $\text{Exp}_{\text{unforgeability}}^{\adv}(\secpar,n,m)$ \\ \hline
        $\Qcert,\Qpres,\Qcorr\exec\emptyset$\\
        $\ctr\exec 0$\\
		$\pp\rexec\ParGen(1^\lambda,n,m)$\\
		$(\isk,\ipk)\rexec\KGen(\pp)$\\
        $(\pi^*,\id_\GG^*,\policy_\node^*,\policy_\edge^*,\id^{\node_\source*},\id^{\node_\dest*},\ell^*,\ctx^*)\rexec\adv^{\Ocert,\Ocorr,\Opres}(\pp,\ipk)$ \\
%%% \Ocert
    \scriptsize
        \quad where $\Ocert$ on input $(\GG,\id_\GG,\{(\attr_i^\node)_{i=1}^n\}_{\node\in\vertices},\{(\attr_j^\edge)_{j=1}^m\}_{\edge\in\edges})$ behaves as follows:\\
    \scriptsize
        \quad\quad $\ctr\exec\ctr+1$\\
    \scriptsize
        \quad\quad $\cert \rexec \Issue(\GG,\id_\GG,\{(\attr_i^\node)_{i=1}^n\}_{\node\in\vertices},\{(\attr_j^\edge)_{j=1}^m\}_{\edge\in\edges},\isk)$\\
    \scriptsize
        \quad\quad $\Qcert \exec \Qcert \cup \{(\cert,\ctr,\GG,\id_\GG,\{(\attr_i^\node)_{i=1}^n\}_{\node\in\vertices},\{(\attr_j^\edge)_{j=1}^m\}_{\edge\in\edges})\}$\\
    \scriptsize
        \quad\quad Return $\ctr$\\
%%% \Opres
    \scriptsize
        \quad where $\Opres$ on input $(\ctr,\policy_\node,\policy_\edge,\id^{\node_\source},\id^{\node_\dest},\ell,\ctx)$ behaves as follows:\\
    \scriptsize
        \quad\quad if $\nexists(\cert,\ctr,\GG,\id_\GG,\{(\attr_i^\node)_{i=1}^n\}_{\node\in\vertices},\{(\attr_j^\edge)_{j=1}^m\}_{\edge\in\edges})\in\Qcert$, return $\bot$\\
    \scriptsize
        \quad\quad if $\ValidPath(\GG,\{(\attr_i^\node)_{i=1}^n\}_{\node\in\vertices},\{(\attr_j^\edge)_{j=1}^m\}_{\edge\in\edges},\policy_\node,\policy_\edge,\id^{\node_\source},\id^{\node_\dest},\ell)=0$, return $\bot$\\
    \scriptsize
        \quad\quad $\pi \rexec \Present(\cert,\ipk,\GG,\id_\GG,\{(\attr_i^\node)_{i=1}^n\}_{\node\in\vertices},\{(\attr_j^\edge)_{j=1}^m\}_{\edge\in\edges},\policy_\node,\policy_\edge,\id^{\node_\source},\id^{\node_\dest},\ell,\ctx)$\\
    \scriptsize
        \quad\quad $\Qpres\exec\Qpres\cup\{(\pi,\ctr,\id_\GG,\policy_\node,\policy_\edge,\id^{\node_\source},\id^{\node_\dest},\ell,\ctx)\}$\\
    \scriptsize
        \quad\quad Return $\pi$\\
%%% \Ocorr
    \scriptsize
        \quad where $\Ocorr$ on input $\ctr$ behaves as follows:\\
    \scriptsize
        \quad\quad if $\nexists(\cert,\ctr,.,.,.,.)\in\Qcert$, return $\bot$\\
     \scriptsize
        \quad\quad $\Qcorr\exec\Qcorr\cup\{\ctr\}$\\
    \scriptsize
        \quad\quad Return $\cert$\\
%%% End oracles
    \normalsize
        Return $1$ if:\\
        \quad $\Verify(\pi^*,\ipk,\id_\GG^*,\policy_\node^*,\policy_\edge^*,\id^{\node_\source*},\id^{\node_\dest*},\ell^*,\ctx^*)=1$, and\\
        \quad $\nexists (.,.,\id_\GG^*,\policy_\node^*,\policy_\edge^*,\id^{\node_\source*},\id^{\node_\dest*},\ell^*,\ctx^*)\in\Qpres$, and\\
        \quad $\nexists (\cert,\ctr,\GG,\id_\GG^*,\{(\attr_i^\node)_{i=1}^n\}_{\node\in\vertices},\{(\attr_j^\edge)_{j=1}^m\}_{\edge\in\edges})\in\Qcert$ such that \\
        \quad\quad\quad $\ctr\in\Qcorr \land 
        \ValidPath(\GG,\{(\attr_i^\node)_{i=1}^n\}_{\node\in\vertices},\{(\attr_j^\edge)_{j=1}^m\}_{\edge\in\edges},\policy_\node^*,\policy_\edge^*,\id^{\node_\source*},\id^{\node_\dest*},\ell^*)=1$\\

        Return $0$\\
		\hline
	\end{tabular}
	\caption{The unforgeability experiment $\text{Exp}_{\text{unforgeability}}^{\adv}(\secpar,n,m)$.}
	\label{fig:unforgeability}
\end{figure}

\subsubsection{Graph Hiding}
For graph hiding, the adversary even chooses the issuer key and both candidate certificates, together with their attributed graphs.
The candidates must use the same graph identifier and satisfy the \emph{same} public request, and both certificates must validate under the chosen issuer key.
The challenger then presents one candidate selected by a hidden random bit.
This follows the canonical indistinguishability-based privacy formulation commonly used for anonymous and attribute-based credentials:
the adversary controls the issuer-side material and provides both candidate witnesses, while the challenge presentation hides which of the two was used.
In particular, no certification oracle is needed, since the adversary already controls the issuer and both candidate certificates and can generate arbitrary auxiliary presentations on either known candidate by itself.
The challenge therefore captures exactly whether the presentation token leaks which certified witness was used.
This rules out a token that reveals, for example, which internal repeater, edge, attribute value, or valid certificate supplied the witness, despite the verifier not receiving any of those objects.

The validity and path checks are essential:
if one candidate could not answer the request, acceptance alone would identify the selected graph.
The game compares presentations for a fixed public request, consistently with deliberately exposing the chosen endpoints, policies, graph identifier, bound, and context.
Padding to the same public $\ell$ ensures that the actual route length need not be exposed by the witness shape.
As in standard credential-privacy games, the definition isolates the privacy of the challenge presentation from information that is already available through the public statement or through independently generated presentations.
For $k\in\{0,1\}$, write $\GG_k=(\vertices_k,\edges_k)$ for the two candidate graphs.

\begin{Definition}[Graph hiding]\label{def:graph_hiding}
A scheme is graph hiding if, for every probabilistic polynomial-time adversary $\adv$ and supported $n,m$, there is a negligible function $\negl$ such that
\[
\Pr\bigl[\mathrm{Exp}_{\mathrm{graph\text{-}hiding}}^{\adv}(\secpar,n,m)=1\bigr]
\leq\frac12+\negl(\secpar).
\]
The full experiment is given in \cref{fig:graph_hiding}.
\end{Definition}

  \begin{figure}[th!]
	\centering\renewcommand{\arraystretch}{1.2}
	\begin{tabular}{|l|}
		\hline
		\rule{0pt}{2.5ex}\textbf{Experiment} $\text{Exp}_{\text{graph-hiding}}^{\adv}(\secpar,n,m)$ \\ \hline
        $b\rexec\bin$\\
        $\pp\rexec\ParGen(1^\lambda,n,m)$\\
        $(\ipk,\id_\GG,(\cert_k,\GG_k,\{(\attr_{ik}^\node)_{i=1}^n\}_{\node\in\vertices_k},\{(\attr_{jk}^\edge)_{j=1}^m\}_{\edge\in\edges_k})_{k=0,1},\policy_\node,\policy_\edge,\id^{\node_\source},\id^{\node_\dest},\ell,\ctx,\state)\rexec\adv(\pp)$ \\
        $\pi^* \rexec \Present(\cert_b,\ipk,\GG_b,\id_\GG,\{(\attr_{ib}^\node)_{i=1}^n\}_{\node\in\vertices_b},\{(\attr_{jb}^\edge)_{j=1}^m\}_{\edge\in\edges_b},\policy_\node,\policy_\edge,\id^{\node_\source},\id^{\node_\dest},\ell,\ctx)$\\
        $b'\rexec\adv(\state,\pi^*)$\\
        Return 0 if:\\
        \quad $\VerifyCert(\cert_k,\GG_k,\id_\GG,\{(\attr_{ik}^\node)_{i=1}^n\}_{\node\in\vertices_k},\{(\attr_{jk}^\edge)_{j=1}^m\}_{\edge\in\edges_k},\ipk)=0$ for some $k=0,1$, or\\
        \quad $\ValidPath(\GG_k,\{(\attr_{ik}^\node)_{i=1}^n\}_{\node\in\vertices_k},\{(\attr_{jk}^\edge)_{j=1}^m\}_{\edge\in\edges_k},\policy_\node,\policy_\edge,\id^{\node_\source},\id^{\node_\dest},\ell)=0$ for some $k=0,1$\\
        Return $[b=b']$\\
		\hline
	\end{tabular}
	\caption{The graph hiding experiment $\text{Exp}_{\text{graph-hiding}}^{\adv}(\secpar,n,m)$.}
	\label{fig:graph_hiding}
\end{figure}

\section{Topology Certification from Multi-Message Signatures}\label{s:abcs}
We construct a topology certificate from signatures on individual nodes and directed edges, and a zero-knowledge proof that selected signed components form a route. This separation lets the issuer certify a complete graph once while a provider later proves a customer-specific path statement without revealing the graph. We give the generic algorithms, state their security with proof sketches, and then describe the concrete BBS/Schnorr proof core. Full generic proofs and algebraic details appear in the supplementary material.

The input graph $\GG'$ has no self-loops, and its node identifiers are unique within that graph. The issuer rejects malformed inputs instead of signing them. A distinguished symbol $\bot$ lies outside every genuine edge-attribute domain, and every supported edge policy satisfies $\policy_\edge(\bot,\ldots,\bot)=0$. Only the issuer adds self-loops, marked with this tuple, to pad shorter routes. A genuine edge cannot exploit the padding branch because its endpoints have distinct identifiers, whereas a padding loop cannot exploit the genuine-edge branch because its attributes fail the edge policy.

\subsection{A Generic Construction}\label{s:generic_construction}
  The construction uses a multi-message signature scheme and a non-interactive proof system for the complete path relation as defined in \cref{s:building_blocks}. Their algorithms carry subscripts $\MMS$ and $\ZKP$. The issuer and provider may work with the complete graph, but the verifier sees only the issuer's public keys and the request-specific proof.

\begin{description}
    \item[Parameter generation.]
      On input the security parameter $1^\secpar$, the number $n$ of node attributes, and the number $m$ of edge attributes, the system parameters are generated as the concatenation of parameters of the building blocks:
      $$
        \pp_{\MMS_\edge}\rexec\ParGen_{\MMS_\edge}(1^\secpar,m+4),\ 
        \pp_{\MMS_\node}\rexec\ParGen_{\MMS_\node}(1^\secpar,n+3),\ \text{and}\ 
        \pp_\ZKP\rexec\ParGen_\ZKP(1^\secpar).
      $$
      The system parameters are $\pp\rexec(\pp_{\MMS_\node},\pp_{\MMS_\edge},\pp_\ZKP)$. A signed node vector has $n+3$ entries: graph identifier, binding tag, node identifier, and $n$ attributes. An edge vector has $m+4$ entries because it includes both endpoint identifiers. Keeping the vector domains distinct also prevents a node signature from being interpreted as an edge signature.
\end{description}

\begin{description}
    \item[Key generation.]
      The issuer generates two key pairs of a multi-message signature scheme supporting efficient protocols as 
      \begin{align*}
        (\isk_\edge,\ipk_\edge) \rexec\KGen_{\MMS_\edge}(\pp_{\MMS_\edge})\quad\text{and}\quad
        (\isk_\node,\ipk_\node) \rexec\KGen_{\MMS_\node}(\pp_{\MMS_\node})\,.
      \end{align*}

      It sets $(\ipk,\isk)\rexec((\ipk_\edge,\ipk_\node),(\isk_\edge,\isk_\node))$. Both verification keys are public, and neither signing key is needed when the provider creates a presentation.
\end{description}

\begin{description}
    \item[Certification.]
      Certification follows a multi-step approach:
      \begin{enumerate}
        \item
          The issuer checks that the node identifiers of $\GG'$ are pairwise distinct and that $\GG'$ contains no self-loops. It then augments $\GG'$ to a graph $\GG$ by adding one directed loop to every vertex and leaving the rest unchanged.
          
          \emph{Rationale:} This will allow us to pad the path lengths whenever needed by repeatedly adding these loops to the actual path.
        \item 
          For every newly added edge, all attributes are set to a special symbol~$\bot$.
          
          \emph{Rationale:} No properties about loops will be proven during presentation, thereby making sure that padding the path does not break any policies that the actual path would have satisfied.
        \item 
          For every edge in the graph $\GG$ (existing or newly added), two additional attributes $\source$ and $\dest$ are added, corresponding to the identifiers of the source and destination nodes of the edge.

          \emph{Rationale:} This will allow us to prove that a set of edges indeed forms a valid path on the graph.
        \item
          The issuer samples a random  binding tag $\rnd\rexec\bin^{\secpar}$.

          \emph{Rationale:} This tag is needed to uniquely bind together all partial signatures belonging to the same graph in order to avoid mix-and-match attacks from combining, e.g., edges from different graphs.
        \item
          For every node $\node\in\vertices$ with attributes $\attr^\node_i$ ($i\in[n]$), the issuer computes a multi-message signature $\sig_\node$ as follows:
          $$
            \sig_\node\rexec\Sign_{\MMS_\node}(\isk_\node,(\id_\GG,\rnd,\id^\node,\attr^\node_1,\dots,\attr^\node_n))\,.
          $$
        \item
          Similarly, for every edge $\edge\in\edges$ with attributes $\attr^\edge_j$ ($j\in[m]$), the issuer computes a multi-message signature $\sig_\edge$ as follows:
          $$
            \sig_\edge\rexec\Sign_{\MMS_\edge}(\isk_\edge,(\id_\GG,\rnd,\id^\source,\id^\dest,\attr^\edge_1,\dots,\attr^\edge_m))\,.
          $$
        \item 
          The certificate is given by
          $$
            \cert\rexec\left(
            \left\{\sig_\edge\right\}_{\edge\in\edges},
            \left\{\sig_\node\right\}_{\node\in\vertices},
            \rnd
            \right).
          $$
      \end{enumerate}
      The random tag is shared by every component of this \emph{one} certification run. A graph identifier alone would not prevent combining components from distinct runs that happen to use the same identifier. The signature on each edge also binds its ordered endpoint identifiers, rather than relying on an unsigned adjacency list supplied by the prover.
\end{description}

\begin{description}
    \item[Certificate validation.]
      This algorithm performs canonical certificate validation. It rejects repeated node identifiers and any self-loop in the input graph, reconstructs the issuer-added loops, and checks that each node and each augmented edge has a valid signature on its exact identifiers, attributes, graph identifier, and common tag. A missing or mismatched component causes rejection. The provider need not transmit this complete graph to customers: their later $\Verify$ call checks a proof about selected components instead.
\end{description}

\begin{description}
    \item[Presentation.]
      Presentation is again a multi-step process:
      \begin{enumerate}
          \item 
            In a first step, the prover determines the necessary path from the source to the destination node, containing the necessary nodes and edges.
            Given a route with at most $\ell$ node positions (at most $\ell-1$ genuine edges), the prover adds as many copies of the loop on the destination node as needed to obtain exactly $\ell$ node positions.
            At the end, it obtains a path
            $$
              (\node_S,\edge_1,\node_2,\edge_2,\dots,\node_D,\selfloop_D,\node_D,\dots,\node_D)\,,
            $$
            which in the following we will refer to as
            $$
              (\node_1,\edge_1,\dots\node_{\ell-1},\edge_{\ell-1},\node_\ell).
            $$
          \item 
            Write $\mathbf a_k=(\attr^{\node_k}_i)_{i=1}^n$ for the attributes of node position $k$ and $\mathbf b_l=(\attr^{\edge_l}_j)_{j=1}^m$ for those of edge position $l$. All intermediate identifiers, attributes, signatures, and the tag belong to the private witness. The proof checks every signature, every node policy, the edge-policy-or-padding condition, and the requested endpoints:
            
            \begin{align*}
                \pi\rexec&\NIZK_\ZKP\bigg[\bigg(%
                \left(\sig_{\node_k},\mathbf a_k,\id^{\node_k}\right)_{k=1}^\ell,
                \left(\sig_{\edge_l},\mathbf b_l\right)_{l=1}^{\ell-1}, \rnd\bigg):\\
                &\bigwedge_{k=1}^{\ell}\bigg[\Verify_{\MMS_\node}
                  \left(\ipk_\node,\left(\id_\GG,\rnd,\id^{\node_k},\mathbf a_k\right),\sig_{\node_k}\right)=1\,\land\,
                \policy_\node(\mathbf a_k)=1\bigg] \,\land\\
                &\bigwedge_{l=1}^{\ell-1}\bigg[\Verify_{\MMS_\edge}
                  \left(\ipk_\edge,\left(\id_\GG,\rnd,\id^{\node_l},\id^{\node_{l+1}},\mathbf b_l\right),\sig_{\edge_l}\right)=1\,\land\,\left(\policy_\edge(\mathbf b_l)=1\,\lor\,\id^{\node_l}=\id^{\node_{l+1}}\right)\bigg]\,\land\\
                &\id^{\node_1}=\id^{\node_S}\,\land\,\id^{\node_\ell}=\id^{\node_D}\bigg](\ctx)\,.
            \end{align*}
            The edge signature is verified on the identifiers of the \emph{adjacent signed nodes}. Because node identifiers are unique, these equality checks prove that the components actually meet. Each edge position must either satisfy the public edge policy or repeat the same node identifier. Since only the issuer-added loops have equal endpoints, removing those positions leaves a genuine route in the audited graph. The first and last identifiers are public, whereas the internal positions, and therefore the actual number of genuine edges, remain hidden beneath the bound $\ell$.

            For example, a route $S\to R\to D$ has three node positions. Under a public bound $\ell=5$, the prover uses the private sequence $S,R,D,D,D$, with two signed loops on $D$ after the two genuine links. All five node positions meet the node policy. The first two edge positions meet the link policy, and the last two use the padding branch. If a longer compliant route exists, it can answer the same public request. The customer sees the bound and the proof, not which sequence was used.
          \item 
            The prover outputs $\pi$.
        \end{enumerate}
\end{description}

\begin{description}
    \item[Verification.]
      The verifier checks $\pi$ against both issuer public keys, the authorized graph identifier, the endpoint identifiers, node and edge policies, the public bound $\ell$, and the context $\ctx$. It accepts if and only if the proof verifies for this complete statement, without requiring the certificate or the route.
\end{description}

For a previously selected route, the relation has exactly $\ell$ node-signature checks and $\ell-1$ edge-signature checks, together with policy and adjacency checks for those positions. This explains why the proof core scales with the padded route and its attributes rather than the total number of network components. Issuing all signatures and discovering a route still depend on the size and structure of the graph.

\subsection{Security of the Generic Construction}\label{s:generic_proof}
The security statements below concern the generic construction instantiated with a proof system for the \emph{complete} path relation. Complete game-based proofs are in \cref{s:supplement_generic_proofs}.

\begin{Theorem}\label{thm:unforgeability}
The generic construction is unforgeable according to \cref{def:unforgeability} if the node and edge multi-message signatures are EUF-CMA secure and the non-interactive proof system is zero-knowledge and fresh-statement weakly simulation extractable.
\end{Theorem}
\paragraph{Proof sketch.}
Replace presentation-oracle replies by simulated proofs, and abort if the issuer reuses a graph binding tag, an event of negligible probability. Fresh-statement weak simulation extractability then yields a witness from a successful presentation for a fresh full public statement. A fresh tag, or a tag of an uncorrupted certificate whose signatures were generated lazily, gives a new signed node vector. If the tag belongs to a corrupted certificate, signatures already issued for that certificate would authenticate a compliant route after removing padding loops, contradicting the winning condition. Thus some extracted node or edge signature is on a new vector, violating EUF-CMA. The complete argument, including oracle simulation and the tag-collision bound, appears in \cref{s:supplement_unforgeability_proof}.\qed

\begin{Theorem}\label{thm:graph-hiding}
The generic construction is graph hiding according to \cref{def:graph_hiding} if the non-interactive proof system is zero-knowledge.
\end{Theorem}
\paragraph{Proof sketch.}
For the two valid challenge graphs, the issuer key and complete public request coincide. Replace real proof parameters and then the challenge token by their simulated counterparts. The simulated token depends only on the public statement, so it is independent of the challenge graph and certificate. The adversary's remaining guessing probability is $1/2$. See \cref{s:supplement_graph_hiding_proof} for the hybrids.\qed

\subsection{BBS Signatures and Composed Proofs}\label{s:instantiation}
We next instantiate the independently signed components with BBS multi-message signatures and linked Schnorr-style representation proofs. 
The complete BBS algorithms and OR-proof transcripts are given in the supplementary material.

\paragraph{Component signatures.}
For each of the two message-vector formats above, the issuer uses a separate BBS key pair~\cite{DBLP:conf/crypto/BonehBS04,DBLP:conf/eurocrypt/TessaroZ23a}. In prime-order pairing groups, a key $\ipk=(X_1,X_2)=(g_1^x,g_2^x)$ verifies a signature $(A,e)$ on $\vec m=(m_1,\ldots,m_L)$ when
\[
  \e(A,X_2\cdot g_2^e)
  =\e\!\left(g_1\prod_{i=1}^{L}h_i^{m_i},g_2\right).
\]
Therefore, the issuer forms $B=g_1\prod_i h_i^{m_i}$ and samples $e\neq-x$ before computing $A=B^{1/(x+e)}$. The graph identifier, common tag, identifiers, and attributes are \emph{signed messages}, which the provider cannot change during presentation. The two formats respectively have $n+3$ and $m+4$ messages, so signing a newly audited component need not change signatures on other components. The static security model does not by itself establish the security of updates across graph epochs. We discuss their limitations in the extensions section.

\paragraph{Hidden signature and route checks.}
To prove possession of a BBS signature without disclosing it, the prover randomizes $A$ and its signed-message commitment. For fresh nonzero $r_1,r_2$, write $D=B^{r_2}$, $\bar A=A^{r_1r_2}$, and $\bar B=D^{r_1}\bar A^{-e}$. The verifier checks $\e(\bar A,X_2)=\e(\bar B,g_2)$. A Schnorr-style proof links this check to the hidden signed messages through
\[
  \bar B=D^{r_1}\bar A^{-e},\qquad
  g_1=D^{r_2^{-1}}\prod_{i=1}^{L}h_i^{-m_i}.
\]
Across all $\ell$ node and $\ell-1$ edge positions, linked hidden-message proofs reuse the same graph tag, equate signed edge endpoints with the adjacent signed node identifiers, and establish the requested predicates. 
While BBS natively supports selective disclosure, complex predicates can be integrated using standard techniques~\cite{DBLP:conf/crypto/CramerDS94,DBLP:phd/dnb/Krenn12,DBLP:conf/asiacrypt/Lipmaa03}.
For the complete path relation, the Fiat-Shamir transcript must bind the issuer keys, graph identifier, endpoints, both policies, $\ell$, context, and all subproof commitments~\cite{DBLP:journals/joc/Schnorr91,DBLP:conf/crypto/FiatS86}.

\paragraph{Padding disjunction.}
For each edge position, the provider must prove $\policy_\edge(\mathbf b_l)=1$ or $\id^{\node_l}=\id^{\node_{l+1}}$.
A Schnorr OR composition~\cite{DBLP:conf/crypto/CramerDS94} proves the branch for which the provider has a witness and simulates the other.
The verifier checks both branch transcripts and that their challenges add to one Fiat-Shamir challenge.
The challenge includes a domain-separated encoding of the \emph{full} public request and all commitments.
For values authenticated as part of a valid certified edge, the two branches are mutually exclusive:
genuine edges have different endpoints, whereas issuer-added loops carry $(\bot,\ldots,\bot)$ and therefore fail every supported edge policy.
As shown in \cref{sec:supplement-instantiation}, special soundness of the underlying protocols is preserved by the OR composition.
Assuming that the used policy-proofs are special honest verifier zero-knowledge and special sound, the entire interactive presentation protocol possesses the corresponding properties.
Since the unforgeability experiment requires that the final public statement be new with respect to all simulated presentation queries, the standard Fiat--Shamir forking argument yields fresh-statement weak simulation extractability in the random-oracle model.
The resulting Fiat-Shamir proof is therefore zero knowledge and fresh-statement weakly simulation extractable as required by \cref{thm:unforgeability}.

%\paragraph{Security boundary.}
%The benchmarks measure this BBS/Schnorr proof core. Applying \cref{thm:unforgeability} to it requires the \emph{whole composed Fiat-Shamir protocol}, including its OR proof and policy gadgets, to satisfy the extraction property in \cref{def:sse} for fresh full public statements. We do not establish that property here. Encrypting a witness and proving ciphertext consistency is a possible augmentation, but it has not been specified, implemented, or included in the reported costs. Thus the generic theorem remains conditional and the measured core is not claimed to be a fully proven unforgeable instantiation. The supplementary material records the detailed algebra and the remaining extraction obligation.

\subsection{Extensions}
In the following we sketch some possible extensions of our approach.
Full details are given in \cref{sec:extension-details}.

\paragraph{Disjoint Paths}\label{sec:multipath}
Distributing a key across node-disjoint routes reduces reliance on individual trusted repeaters~\cite{acns2026}. An issuer may also sign a stable, high-entropy secret $x_\node$ with each node, retaining the same secret for every valid certification of that node in a graph epoch. For a presentation-specific context $\ctx$, the prover derives $\mathsf{nym}_{\node,\ctx}=H(\ctx)^{x_\node}$ and proves that the exponent is the one in the certified node attributes. Equal intermediate nodes then yield equal pseudonyms under the same context, so the verifier can reject intersecting paths without learning their identifiers~\cite{DBLP:conf/eurocrypt/NaorPR99,DBLP:books/daglib/p/BichselCDEKLNP15}.

When routes are padded, repeated destination-node occurrences must be excluded from the comparison. The proof therefore distinguishes a genuine node with a certified pseudonym from a padding repetition, using a disjunction tied to the same signed path. 

\paragraph{Changes to the Certified Graph}
In contrast to \cite{DBLP:conf/ccs/Gross14}, new nodes or links may be certified with the existing graph identifier and binding tag without resigning unchanged components, provided that added nodes receive identifiers unused in that graph epoch. Removing certified objects requires a new graph epoch with a fresh identifier and a way for verifiers to reject presentations based on superseded epochs because old signatures otherwise remain usable. 

\paragraph{Non-connectivity}
Our construction also supports non-connectivity proofs and thus provides the same core graph-certification functionality as approaches such as~\cite{DBLP:conf/ccs/Gross14}, namely proofs of connectivity and separation.
To do so, the issuer additionally certifies a connected-component identifier for each node.
To prove that two nodes are disconnected, the prover shows valid signatures on both nodes and proves that their component identifiers differ.
This yields a constant-size proof, independent of the graph size.
If an update merges two components, the affected nodes have to be re-certified.

\section{Evaluation}\label{s:evaluation}
We next evaluate the practical costs of our construction along three dimensions:
certificate size, presentation size, and presentation runtime.

For the presentation-size and runtime evaluations, we consider equality-based attribute predicates as a representative and easily quantifiable class of policies.
Such a predicate fixes an attribute to a public value and therefore has, from the verifier's perspective, essentially the same information-disclosure effect as revealing that attribute.
Accordingly, the disclosure rates used below denote the fraction of attributes constrained to public values, rather than attributes that are necessarily transmitted in the clear at every path position.

Our estimates and measurements focus on the cryptographic presentation layer.
They do not include route discovery or other graph-management operations, which may naturally depend on the size and structure of the underlying QKD network.

\subsection{Certificate Size}
Starting from an undirected graph $\GG=(\vertices,\edges)$, the transformation used by our construction preserves the $|\vertices|$ nodes, replaces every undirected edge by two directed edges, and adds one self-loop to every node.
The resulting augmented graph therefore contains $|\vertices|$ nodes and $2|\edges|+|\vertices|$ directed edges, yielding a total of $2|\vertices|+2|\edges|$ certified graph objects.
This results in a certificate size $|\cert|$ of 
$$
    |\cert|=|\vertices|\cdot|\sig_\node|+(2|\edges|+|\vertices|)\cdot|\sig_\edge|+\secpar\text{ bits}\ ,
$$
where $|\sig_\node|$ and $|\sig_\edge|$ are the signature sizes of the respective schemes, and the $\secpar$ bits come from the binding tag $\rnd$. Storage of the graph, attributes, identifiers, and indexes needed for route selection is additional. Issuance, certificate validation, and route search can scale with $|\vertices|+|\edges|$.
For the BBS signature $(A,e)$ described in the supplementary material, compressed BLS12-381 encodings give $48+32=80$ bytes per signature, or $160(|\vertices|+|\edges|)$ bytes of signatures in this undirected-graph transformation. This estimate covers signatures only, with message vectors, graph data, encoding overhead, and the binding tag adding further storage.

\subsection{Presentation Size}
In the following we estimate the overall proof size in our construction. A padded route uses $\ell$ node proofs and $\ell-1$ edge proofs, each with a genuine-edge/padding OR proof:
\[
    |\pi|=\ell|\pi_\node|+(\ell-1)(|\pi_\edge|+|\pi_{\mathsf{OR}}|).
\]
Here the three terms are the node, edge, and additional OR-proof sizes. At fixed $\ell,n,m$ and predicates, this cryptographic size is independent of graph size.

\iffalse
For the BBS core of \cref{s:instantiation}, a signature-possession proof has randomized elements $\bar A,\bar B,D$, two Schnorr commitments, one response per hidden message, and three responses for $r_1,r_3,e$. Assume $50\%$ disclosure of node and edge attributes. A node proof hides $\rnd$, its identifier, and $n/2$ attributes, costing five $\G_1$ elements and $n/2+5$ scalars. Edge attributes stay hidden in BBS because disclosing them there would reveal padding. The OR branches instead prove selective disclosure or $\id^{\node_l}=\id^{\node_{l+1}}$. The edge BBS proof uses five group elements and $m+6$ scalars. A direct OR proof adds $m/2+1$ equality commitments and responses plus two challenges, for $m/2+6$ group elements and $3m/2+9$ scalars per edge position.
\fi

Inspecting the structure of the presentation proof and counting the resulting group elements and scalar responses gives the following costs.
Assuming that $50\%$ of the node and edge attributes are disclosed, we obtain five $\G_1$ elements and $n/2+5$ scalars per node position, and $m/2+6$ $\G_1$ elements and $3m/2+9$ scalars per edge position.

With $48$-byte compressed $\G_1$ elements and $32$-byte $\Z_p$ scalars on BLS12-381, the counts and size at $50\%$ disclosure are
\begin{align*}
    N_{\G_1}&=5\ell+(\ell-1)(m/2+6),\\
    N_{\Z_p}&=\ell(n/2+5)+(\ell-1)(3m/2+9),\\
    |\pi|_{\mathrm{bytes}}&\approx48N_{\G_1}+32N_{\Z_p}.
\end{align*}
The same counting method gives \cref{tab:proof-sizes} at other disclosure rates.

\begin{table*}[t]
    \centering
    \caption{Presentation sizes in KiB for different path lengths, attribute dimensions, and disclosure rates.}
    \label{tab:proof-sizes}
    \small
    \setlength{\tabcolsep}{3.5pt}
    \renewcommand{\arraystretch}{1.05}
    \begin{tabular}{@{}cc@{\hspace{15pt}} rrr r@{\hspace{10pt}} rrr r@{\hspace{10pt}} rrr@{}}
        \toprule
        & &
        \multicolumn{3}{c}{$20\%$ disclosed} &&
        \multicolumn{3}{c}{$50\%$ disclosed} &&
        \multicolumn{3}{c}{$80\%$ disclosed} \\
        \cmidrule(lr){3-5}
        \cmidrule(lr){7-9}
        \cmidrule(lr){11-13}
        & &
        \multicolumn{3}{c}{$n$} &&
        \multicolumn{3}{c}{$n$} &&
        \multicolumn{3}{c}{$n$} \\
        $\ell$ & $m$
        & $10$ & $20$ & $50$ &
        & $10$ & $20$ & $50$ &
        & $10$ & $20$ & $50$ \\
        \midrule
        \multirow{3}{*}{$5$}
        & $10$ & $7.3$ & $8.6$ & $12.3$ & & $7.8$ & $8.6$ & $10.9$ & & $8.3$ & $8.6$ & $9.5$ \\
        & $20$ & $9.2$ & $10.5$ & $14.2$ & & $10.6$ & $11.4$ & $13.7$ & & $12.0$ & $12.3$ & $13.3$ \\
        & $50$ & $14.8$ & $16.1$ & $19.8$ & & $19.0$ & $19.8$ & $22.2$ & & $23.3$ & $23.6$ & $24.5$ \\
        \midrule
        \multirow{3}{*}{$10$}
        & $10$ & $15.7$ & $18.2$ & $25.7$ & & $16.9$ & $18.4$ & $23.1$ & & $18.0$ & $18.7$ & $20.5$ \\
        & $20$ & $19.9$ & $22.4$ & $29.9$ & & $23.2$ & $24.8$ & $29.4$ & & $26.5$ & $27.1$ & $29.0$ \\
        & $50$ & $32.6$ & $35.1$ & $42.6$ & & $42.2$ & $43.7$ & $48.4$ & & $51.8$ & $52.4$ & $54.3$ \\
        \midrule
        \multirow{3}{*}{$50$}
        & $10$ & $82.6$ & $95.1$ & $132.6$ & & $89.4$ & $97.2$ & $120.6$ & & $96.2$ & $99.3$ & $108.7$ \\
        & $20$ & $105.5$ & $118.0$ & $155.5$ & & $123.8$ & $131.6$ & $155.1$ & & $142.1$ & $145.2$ & $154.6$ \\
        & $50$ & $174.4$ & $186.9$ & $224.4$ & & $227.2$ & $235.0$ & $258.4$ & & $279.9$ & $283.0$ & $292.4$ \\
        \bottomrule
    \end{tabular}
\end{table*}

Looking at the numbers, different trends become visible depending on the relative numbers of node and edge attributes.
If $m$ is small compared with $n$, increasing the disclosure rate tends to reduce the overall presentation size.
As $m$ increases, this trend is reversed, and higher disclosure rates can instead lead to larger presentations.
The effect also depends slightly on the path length $\ell$, since a path contains $\ell$ node proofs but only $\ell-1$ edge proofs.
As $\ell$ increases, this asymmetry becomes less relevant and the observed behavior stabilizes.

A closer look at the dependence of the proof size on the proof goal explains these trends.
For node proofs, every disclosed attribute becomes public and therefore removes one response from the underlying BBS proof.
Consequently, a higher disclosure rate directly reduces the size of the node proof, as is common in anonymous credential systems.
For edge proofs, however, attributes cannot simply be disclosed as public BBS messages, since this would distinguish genuine edges from padding loops.
Instead, selective disclosure has to be established inside the corresponding branch of the OR proof.
Hence, every additionally disclosed edge attribute introduces additional proof material.
The overall effect of a higher disclosure rate is therefore determined by the trade-off between the savings in the node proofs and the additional cost of the edge-policy proofs.

Importantly, the resulting presentation sizes are independent of the size of the underlying QKD network.
This contrasts with previous approaches such as~\cite{qcnc2026,thesis-margherita}, where the proof size grows with the size of the certified graph.
This property is particularly beneficial for large deployments, since adding further nodes or links does not increase the communication cost of proving properties of a path with fixed length and attribute dimensions.

Even for the largest parameter setting considered, with $\ell=m=n=50$, the resulting presentation size remains below $300$\,KiB.
To put this into perspective, assuming QKD links of up to approximately $100$\,km, $\ell=50$ node positions correspond to an end-to-end span of almost $5'000$\,km.
This exceeds the geographical distance from Washington, D.C.\ to any state capital in the contiguous United States, which is approximately $3{,}800$\,km for the most distant capitals on the West Coast.
Similarly, within continental Europe, the distance from Brussels to even the most distant EU capitals is around $2{,}100$\,km, with Athens representing one of the extreme cases.
Thus, evaluated path lengths exceed plausible continent-scale distances, while the presentation sizes remain practical for real-world use.

\subsection{Runtime}\label{s:runtime}
We benchmarked our Rust implementation on Ubuntu~24.04.4 LTS (Linux~7.0.0-31-generic), an Intel Core Ultra 7 265U (up to $5.4$\,GHz), and $32$\,GB RAM. Both experiments selectively disclose $50\%$ of node and edge attributes.\footnote{The implementation can be accessed at \url{https://github.com/alecolo129/qkd-topology-cert}.}

At $m=n=8$, increasing padded length $\ell$ from $2$ to $64$ raises generation from $7.2$ to $353$\,ms and verification from $8.5$ to $275$\,ms (\cref{fig:runtime-path-length}). Growth is approximately linear because each position adds constant proof work at fixed attribute counts.

\begin{figure}[t]
    \centering
    \begin{minipage}[t]{0.49\linewidth}
        \centering
        \includegraphics[width=0.95\linewidth]{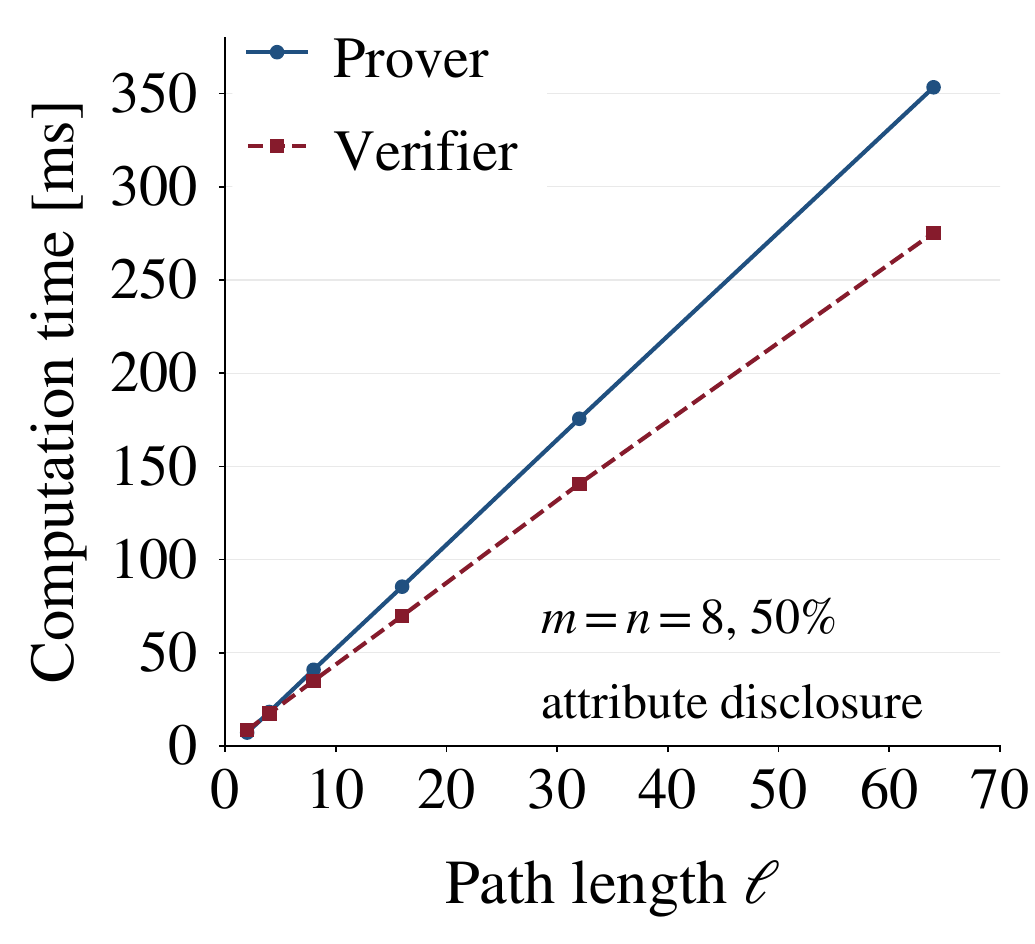}
        \caption{Runtime versus path length $\ell$ with $m=n=8$ and $50\%$ disclosure.}
        \label{fig:runtime-path-length}
    \end{minipage}\hfill
    \begin{minipage}[t]{0.49\linewidth}
        \centering
        \includegraphics[width=0.95\linewidth]{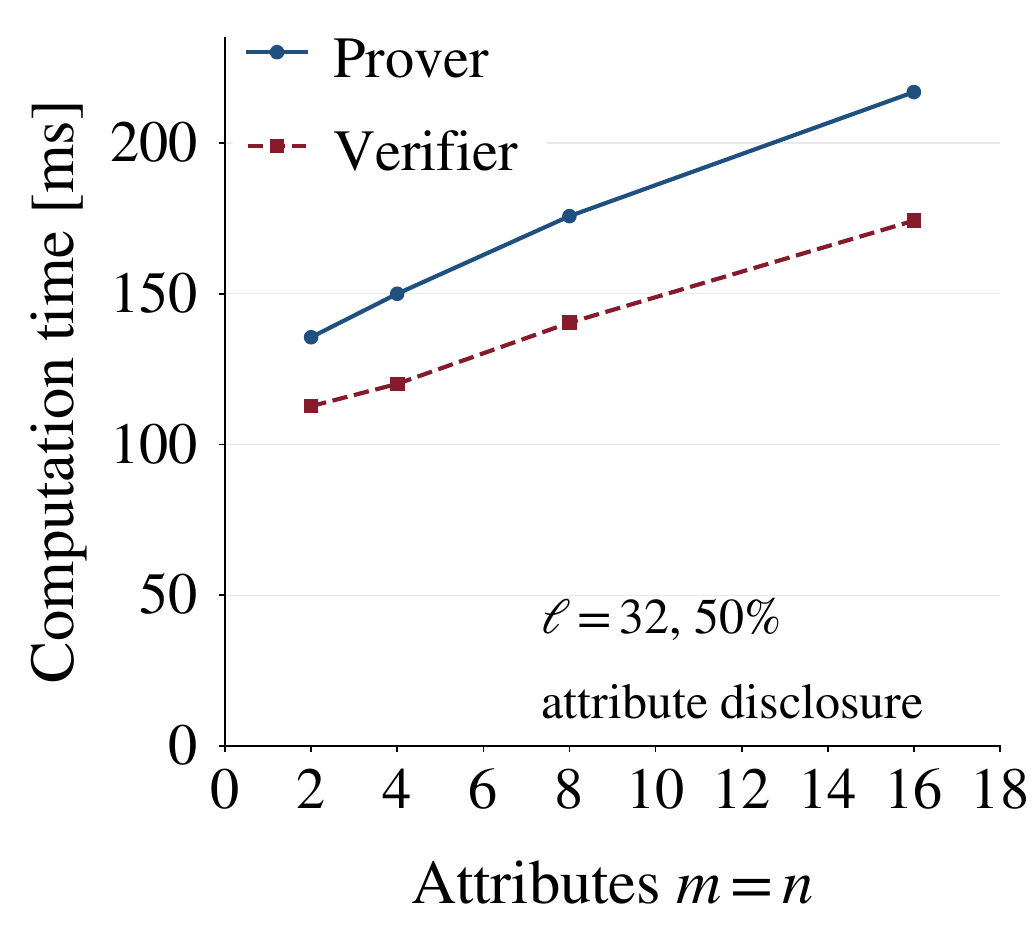}
        \caption{Runtime versus attributes $m=n$ with $\ell=32$ and $50\%$ disclosure.}
        \label{fig:runtime-attributes}
    \end{minipage}
\end{figure}

At $\ell=32$, increasing $m=n$ from $2$ to $16$ raises generation from $136$ to $217$\,ms and verification from $113$ to $174$\,ms (\cref{fig:runtime-attributes}), consistent with roughly linear hidden-attribute proof work. These experiments vary the padded path length and the number of attributes.
Varying the overall graph size is unnecessary, since the presentation proofs are independent of nodes and edges that do not occur on the presented path.

%The supplied materials have mean curves but no raw samples, confidence intervals, benchmark configuration, implementation commit, or matched baseline. The numerical comparison remains descriptive pending reproducible, matched experiments.

\section{Conclusions}\label{s:conclusion}

We presented a topology-certification mechanism for QKD networks that allows providers to prove the existence of policy-compliant routes without revealing the routes in an individual presentation.
By certifying nodes and edges independently, presentation size and cryptographic proof work for an already selected route depend on that route rather than on the size of the complete network. Certificate storage and route discovery still depend on the graph, and updates require explicit epoch management.
The construction further supports hidden path lengths and extensions to node-disjoint multi-path proofs and non-connectivity proofs.
The BBS/Schnorr proof core has promising communication and computation costs for the studied parameter settings. 

An important direction for future work is the development of an efficient post-quantum instantiation of our construction.
A further step is end-to-end integration with practical QKD network management and path selection, including an independently enforced binding between the proved existence statement and the route actually used for key material.

\paragraph{Acknowledgements}
This work has received funding from the European Union's Horizon Europe research and innovation programme under Grant Agreement No.~101114043 ({\sc qsnp}), and from the Connecting Europe Facility (CEF Digital) under Grant Agreement No.~101305103 ({\sc QuarterNext}), with national co-funding provided by the Austrian Research Promotion Agency (FFG).

Views and opinions expressed are, however, those of the authors only and do not necessarily reflect those of the European Union or the respective granting authorities. Neither the European Union nor the granting authorities can be held responsible for them.

\bibliographystyle{splncs04}
\bibliography{references}

%\clearpage
\appendix
%\section*{Supplementary Material}
%The following material formalizes the cryptographic building blocks and gives full generic proofs, concrete instantiation details, and extension details.
\section{Formal Proof-System Security Games}\label{s:supplement_proof_games}

This section gives the full versions of the five definitions in \cref{s:building_blocks}, including the algorithms, correctness conditions, and security games. We use the notation of \cref{s:notation}.

\subsection{Cryptographic Building Blocks}\label{s:supp-building-blocks}
In the following we introduce the cryptographic building blocks that will be used in our main constructions.

\subsubsection{Multi-message signatures.}
Intuitively, a multi-message signature scheme allows a signer to authenticate a vector of messages such that a verifier can later check that the signature was issued by the signer on exactly these messages.

\begin{Definition}[Multi-message signature scheme]\label{def:supp-multi-message-signatures}
A multi-message signature scheme consists of a tuple of \PPT algorithms $(\ParGen,\KGen, \Sign, \Verify)$, where:
\begin{description}
\item[Parameter Generation.]
$\ParGen(1^\secpar,L)$ takes as input the security parameter $1^\secpar$ and the supported length $L$ of messages, and outputs public parameters $\pp$.

\item[Key Generation.]
$\KGen(\pp)$ takes as input the public parameters $\pp$ and outputs a signing key $\sk$ and a verification key $\pk$:
$$
  (\sk,\pk) \gets \KGen(\pp)\,.
$$

\item[Sign.]
$\Sign(\sk, \vec{m})$ takes as input the signing key $\sk$ and a message vector $\vec{m}=(m_1,\ldots,m_L)$ with messages $m_i$ from the message space $\mathcal{M}$, and outputs a signature $\sig$:
$$
  \sig \gets \Sign(\sk,\vec{m})\,.
$$

\item[Verify.]
$\Verify(\pk, \vec{m}, \sig)$ takes as input the verification key $\pk$, a message vector $\vec{m}=(m_1,\ldots,m_L)$, and a signature $\sig$, and outputs $1$ if $\sig$ is a valid signature on $\vec{m}$ under $\pk$, and $0$ otherwise.
\end{description}
\end{Definition}
For notational convenience, while being input to all algorithms, we will not make the public parameters explicit when clear from the context.

For correctness, we require that honestly generated signatures always verify correctly.
That is, for all $\secpar,L \in \mathbb{N}$ and all message vectors $\vec{m} \in \mathcal{M}^L$, we have that:
\[
    \Pr\left[\Verify(\pk,\vec{m},\sig)=1
    :
    \begin{array}{c}
    \pp\rexec\ParGen(1^\secpar,L);
    (\sk,\pk)\rexec\KGen(\pp);\\
    \sig\rexec\Sign(\sk,\vec{m})
    \end{array}
    \right] = 1\,.
\]

A signature scheme is unforgeable if no efficient adversary can produce a valid signature on a message vector that was not signed before.
We use the standard notion of existential unforgeability under adaptive chosen-message attacks (EUF-CMA), first introduced by \cite{DBLP:journals/siamcomp/GoldwasserMR88}.

\begin{Definition}[Existential unforgeability under chosen-message attacks]\label{def:supp-euf-cma}
For each supported vector length $L$, a multi-message signature scheme $(\ParGen,\KGen,\Sign,\Verify)$ is \emph{existentially unforgeable under chosen-message attacks (EUF-CMA)} if for every \PPT adversary $\adv$ there exists a negligible function $\negl$ such that
\[
\Pr\left[
\begin{array}{c}
\Verify(\pk,\vec{m}^*,\sig^*)=1 \;\land\; \\
\vec{m}^* \notin Q
\end{array}
:
\begin{array}{c}
\pp\rexec\ParGen(1^\secpar,L); \\
(\sk,\pk)\rexec\KGen(\pp); \\
(\vec{m}^*,\sig^*)\rexec \adv^{\Sign(\sk,\cdot)}(\pp,\pk)
\end{array}
\right]
\leq \negl(\secpar).
\]
Here, $Q$ denotes the set of message vectors queried by $\adv$ to its signing oracle.
\end{Definition}

\subsubsection{Non-Interactive Zero-Knowledge Proofs}
Intuitively, a non-interactive zero-knowledge proof system allows a prover to convince a verifier that a statement belongs to a language without revealing any information about the corresponding witness.
Weak simulation extractability additionally guarantees that, even after seeing arbitrarily many simulated proofs, no efficient adversary can produce a valid proof for a fresh statement without knowing a corresponding witness.

Let $\mathcal{R}$ be a binary relation, and let
$$
  \mathcal{L}_{\mathcal{R}} := \left\{x \,\middle|\, \exists w : \mathcal{R}(x,w)=1 \right\}
$$
be the corresponding language.
For a path presentation, $x$ includes the complete public request, including the context $\ctx$.

\begin{Definition}[Non-interactive zero-knowledge proof system]\label{def:supp-nizk}
A non-interactive zero-knowledge proof system for a relation $\mathcal{R}$ consists of a tuple of \PPT algorithms
\[
  (\ParGen,\NIZK,\Verify),
\]
where:
\begin{description}
\item[Parameter Generation.]
$\ParGen(1^\secpar)$ takes as input the security parameter $1^\secpar$ and outputs public parameters $\pp$.

\item[Prove.]
$\NIZK(\pp,x,w)$ takes as input the public parameters $\pp$, a statement $x$, and a witness $w$ such that $\mathcal{R}(x,w)=1$, and outputs a proof $\pi$:
\[
  \pi \gets \NIZK(\pp,x,w)\,.
\]

\item[Verify.]
$\Verify(\pp,x,\pi)$ takes as input the parameters $\pp$, a statement $x$, and a proof $\pi$, and outputs $1$ if $\pi$ is a valid proof for $x$, and $0$ otherwise.
\end{description}
\end{Definition}

For notational convenience, we sometimes omit the public parameters when clear from context.

For correctness, we require that honestly generated proofs always verify correctly.
That is, for all $\secpar\in\mathbb{N}$ and all $(x,w)$ such that $\mathcal{R}(x,w)=1$, we have that
\[
  \Pr\left[
    \Verify(\pp,x,\pi)=1
    :
    \begin{array}{c}
      \pp\rexec\ParGen(1^\secpar);\\
      \pi\rexec\NIZK(\pp,x,w)
    \end{array}
  \right]
  =1\,.
\]

A non-interactive proof system is zero-knowledge if there exists a simulator that can generate public parameters together with a simulation trapdoor and subsequently simulate proofs without knowing the corresponding witnesses.
We use the standard notion of zero knowledge---originally defined by \cite{DBLP:conf/stoc/BlumFM88}---following the notation of~\cite{DBLP:conf/cans/BobolzEKRS21}.

\begin{Definition}[Zero knowledge]\label{def:supp-zero-knowledge}
A non-interactive proof system $(\ParGen,\NIZK,\Verify)$ for a relation $\mathcal{R}$ is \emph{zero-knowledge} if there exist \PPT algorithms $(\Sim_1,\Sim_2)$ such that for every \PPT adversary $\adv$ there exists a negligible function $\negl$ satisfying
\[
  \left|
    \Pr\left[
      b'=1
      :
      \begin{array}{c}
        \pp\rexec\ParGen(1^\secpar);\\
        (x,w,\state)\gets\adv(\pp);\\
        \pi\rexec\NIZK(\pp,x,w);\\
        b'\gets\adv(\state,\pi)
      \end{array}
    \right]
    -
    \Pr\left[
      b'=1
      :
      \begin{array}{c}
        (\pp,\tau_{\mathsf{sim}})\rexec\Sim_1(1^\secpar);\\
        (x,w,\state)\gets\adv(\pp);\\
        \pi\rexec\Sim_2(\tau_{\mathsf{sim}},x);\\
        b'\gets\adv(\state,\pi)
      \end{array}
    \right]
  \right|
  \leq
  \negl(\secpar),
\]
where the adversary is restricted to output $(x,w)$ such that $\mathcal{R}(x,w)=1$.
\end{Definition}

Following the weak simulation-extractability framework of Faust et al.~\cite{DBLP:conf/indocrypt/FaustKMV12}, we define the following variant which forces the adversary's final statement to be fresh.

\begin{Definition}[Fresh-statement weak simulation extractability]
\label{def:supp-wse}
A non-interactive proof system $(\ParGen,\allowbreak\NIZK,\allowbreak\Verify)$ for a relation $\mathcal{R}$ in the random-oracle model is \emph{fresh-statement weakly simulation extractable} if, for every \PPT adversary $\adv$, there exists an efficient extractor $\Ext^{\adv}$ that may rewind $\adv$ and has access to its random tape as well as to the random-oracle and proof-simulation transcripts.

Let $Q$ denote the set of statements submitted by $\adv$ to the proof-simulation oracle.
Define
$$
    \mathsf{acc}
    :=
    \Pr\left[
        \begin{array}{c}
            \Verify(\pp,x^*,\pi^*)=1
            \;\land\;
            x^*\notin Q
        \end{array}
        :
        \begin{array}{c}
            (\pp,\tau_{\mathsf{sim}})
            \rexec
            \Sim_1(1^\secpar);
            \\
            (x^*,\pi^*)
            \gets
            \adv^{\Sim_2(\tau_{\mathsf{sim}},\cdot),H}(\pp;\rho)
        \end{array}
    \right]
$$
and
$$
    \mathsf{ext}
    :=
    \Pr\left[
        \begin{array}{c}
            \mathcal{R}(x^*,w^*)=1
            \;\land\;
            x^*\notin Q
        \end{array}
        :
        \begin{array}{c}
            (\pp,\tau_{\mathsf{sim}})
            \rexec
            \Sim_1(1^\secpar);
            \\
            (x^*,\pi^*)
            \gets
            \adv^{\Sim_2(\tau_{\mathsf{sim}},\cdot),H}(\pp;\rho);
            \\
            w^*
            \gets
            \Ext^{\adv}(x^*,\pi^*;\rho,T_H,T)
        \end{array}
    \right],
$$
where $T_H$ and $T$ denote the random-oracle and proof-simulation transcripts, respectively.
There exist a constant $d>0$, a polynomial $p$, and a negligible extraction error $\nu$ such that, whenever $\mathsf{acc}\geq\nu$,
$$
    \mathsf{ext}
    \geq
    \frac{1}{p(\secpar)}
    \bigl(\mathsf{acc}-\nu\bigr)^d.
$$
\end{Definition}
In our unforgeability experiment, the final statement $x^*$ must be fresh with respect to all queries to the presentation oracle.
Since these queries are answered by the proof simulator in Game~1, this gives exactly $x^*\notin Q$, as required by the above definition.
\section{Full Proofs of the Generic Construction}\label{s:supplement_generic_proofs}
This section proves the generic theorems stated in \cref{s:generic_proof}. The concrete BBS/Schnorr proof core described later requires a separate proof that its full composition meets the assumed extractability property.

\subsection{Unforgeability}\label{s:supplement_unforgeability_proof}
We prove \cref{thm:unforgeability} against the experiment in \cref{fig:unforgeability}.
\begin{proof}
We define a series of games.

\textbf{Game 0.}
This is the original unforgeability experiment from \cref{fig:unforgeability}.

\textbf{Game 1.}
We modify the generation of the proof-system parameters.
Instead of computing $\pp_{\ZKP} \rexec \ParGen_{\ZKP}(1^\secpar)$, we compute $(\pp_{\ZKP},\tau_{\mathsf{sim}}) \rexec \Sim_1(1^\secpar)$, and use the simulated parameters $\pp_{\ZKP}$ as the public parameters of the
proof system.

Furthermore, all answers to presentation queries are generated using the simulation trapdoor. 
More precisely, on a query
$$
  (\ctr,\phi_\node,\phi_\edge, \id^{\node_\source},\id^{\node_\dest},\ell,\ctx),
$$
the oracle first performs the same validity checks as in the original
experiment and then computes
$$
  \pi
  \rexec
  \Sim_2\left( \tau_{\mathsf{sim}}, \ipk, \id_\GG, \phi_\node, \phi_\edge, \id^{\node_\source}, \id^{\node_\dest}, \ell, \ctx \right).
$$

By the assumed zero-knowledge property, Game~0 and Game~1 are computationally indistinguishable.

\textbf{Game 2.}
We abort if two distinct certification queries are assigned the same binding
tag.
Let $q_{\cert}$ denote the number of certification queries made by the
adversary.
Since every binding tag is sampled independently and uniformly from
$\{0,1\}^{\secpar}$, the probability of an abort is bounded by
$$
  \Pr[\mathsf{abort}] \leq \frac{q_{\cert}^2}{2^\secpar}.
$$
Since $q_\cert$ is polynomially bounded by the runtime of $\adv$, Game~1 and Game~2 differ only with negligible probability.

\textbf{Game 3.}
Assume that the adversary wins Game~2 with non-negligible probability.
Consider the adversary together with the simulated presentation oracle and the random oracle as an adversary against fresh-statement weak simulation extractability.
By the winning condition of \cref{fig:unforgeability}, the public statement
$$
    x^* :=
    \bigl(
        \ipk,
        \id_\GG^*,
        \phi_\node^*,
        \phi_\edge^*,
        \id^{\node_S^*},
        \id^{\node_D^*},
        \ell^*,
        \ctx^*
    \bigr)
$$
was not previously submitted to the presentation oracle.
Since all presentation-oracle replies in Game~1 are generated by the proof simulator, we have $x^*\notin Q$, where $Q$ is the set of statements submitted to the proof-simulation oracle.
Fresh-statement weak simulation extractability therefore gives an efficient, possibly rewinding extractor that, whenever the adversary produces an accepting proof $\pi^*$ for $x^*$ with non-negligible probability, extracts with non-negligible probability a witness $w^*$ satisfying $\mathcal{R}(x^*,w^*)=1$.
The reduction fixes its internal randomness across rewinding executions and answers repeated signing queries consistently by reusing previously obtained signing-oracle responses.
In particular, the extracted witness contains a binding tag $\rnd^*$, signatures and attributes for all nodes and edges of the claimed path, and the corresponding node identifiers, such that all signature verifications and policy checks in the proof relation are satisfied.

\medskip

It remains to show that a successful adversary in Game~3 yields an EUF-CMA forgery against one of the two deployed multi-message signature schemes.
Consider therefore the binding tag $\rnd^*$ contained in the extracted witness, which was used in at most one certification query, since Game~2 did not abort.

We now distinguish the following cases:

\begin{description}
    \item[Fresh $\rnd^*$.] 
      Suppose that no certification query was assigned the tag $\rnd^*$.
      Then every extracted node signature is valid for a message vector that was never signed during certification.
      In particular, any extracted node signature constitutes an EUF-CMA forgery against the node-signature scheme.
    \item[Uncorrupted $\rnd^*$.]
      Next, suppose that $\rnd^*$ belongs to a certification query with identifier $\ctr^*$, but that $\ctr^*\notin\Qcorr$. 
      In the reduction we now perform certificate generation lazily:
      on a certification query, the reduction samples and records the binding tag and all graph data, but generates the individual signatures only if the
corresponding certificate is later corrupted.
     Until corruption, all presentation queries for uncorrupted certificates are answered using the proof simulator and therefore do not require the underlying signatures.
     Consequently, no signature containing the tag $\rnd^*$ has been requested from the signing oracle.
     Thus, any node signature extracted from $\pi^*$ again constitutes an EUF-CMA forgery against the node-signature scheme.
    \item[Corrupted $\rnd^*$.]
      Finally, suppose that $\rnd^*$ belongs to a corrupted certificate with identifier $\ctr^*\in\Qcorr$.
      By the winning condition of \cref{fig:unforgeability}, the graph and attributes associated with this certificate now cannot satisfy
      $$
        \ValidPath\bigl( \GG, \{\attr^\node\}, \{\attr^\edge\}, \phi_\node^*, \phi_\edge^*, \id^{\node_\source^*}, \id^{\node_\dest^*}, \ell^* \bigr)=1.
      $$
      Assume that all signatures contained in the extracted witness were signatures issued as part of this corrupted certificate.
      Because the binding tag, graph identifier, node identifiers, edge endpoints, and attributes are included in the signed message vectors, the extracted signatures then authenticate nodes and edges of the certified augmented graph. Uniqueness of node identifiers within that graph makes equality of adjacent signed endpoint identifiers imply equality of the actual graph vertices.
      Moreover, the proof relation guarantees that these edges form a path from $\id^{\node_\source^*}$ to $\id^{\node_\dest^*}$ and that all non-padding nodes and edges satisfy the respective policies.
      Removing the padding loops from this route yields a compliant route in the original graph with at most $\ell^*$ node positions (and at most $\ell^*-1$ genuine edges), contradicting the winning condition of the unforgeability experiment.

      Thus, at least one signature contained in the extracted witness is valid for a node or edge message vector that was not signed as part of the corrupted certificate.
      Since binding tags do not collide in Game~2, this message vector was not signed as part of any other certificate either.
      Hence, this signature is an EUF-CMA forgery against either the node-signature scheme or the edge-signature scheme.

\end{description}

A reduction now guesses whether the resulting forgery is a node or an edge signature and embeds its EUF-CMA challenge public key into the corresponding issuer public key.
It generates the other signing key honestly, answers corruption queries using its signing oracle, answers presentation queries using the proof simulator, and outputs the extracted fresh signature.
The guess incurs only a multiplicative constant loss of $\nicefrac{1}{2}$.

We conclude that a non-negligible probability of winning Game~3 implies either a violation of fresh-statement weak simulation extractability or a non-negligible EUF-CMA advantage against one of the deployed multi-message signature schemes.
Together with the negligible differences introduced in Games~1 and~2, the claim of the theorem follows.\qed
\end{proof}

\subsection{Graph Hiding}\label{s:supplement_graph_hiding_proof}
We prove \cref{thm:graph-hiding} against the experiment in \cref{fig:graph_hiding}.
\begin{proof}
  We define a series of games to prove the theorem.

  \textbf{Game 0.}
  This is the original graph hiding experiment as in \cref{def:graph_hiding}.

  \textbf{Game 1.}
  We replace the proof-system setup with
  \[
    (\pp_\ZKP',\tau_{\mathsf{sim}})\rexec\Sim_1(1^\secpar)
  \]
  and publish the simulated parameters $\pp_\ZKP'$ in the system parameters.

  By the zero-knowledge property of the proof system, Game 0 and Game 1 are computationally indistinguishable.

  \textbf{Game 2.}
  We compute the challenge proof with the simulator, using its simulation trapdoor and the full public statement:
  \[
    \pi'\rexec\Sim_2\!\left(\tau_{\mathsf{sim}},
      (\ipk,\id_\GG,\policy_\node,\policy_\edge,
       \id^{\node_\source},\id^{\node_\dest},\ell,\ctx)\right).
  \]

  Again by the zero-knowledge property of the proof system, Game 1 and Game 2 are computationally indistinguishable.

  \medskip

  The distribution of $\pi'$ now depends only on the public statement
  \[
    \big(\ipk,\id_\GG,\policy_\node,\policy_\edge,
       \id^{\node_\source},\id^{\node_\dest},\ell,\ctx\big),
  \]
  which is identical for both challenge instances.
  In particular, $\pi'$ is independent of the challenge bit $b$ and of the certificate, graph, signatures, binding tag, and node and edge attributes selected by $b$.
  Consequently, the adversary guesses $b$ correctly with probability exactly $1/2$.

  Now, since all games are computationally indistinguishable, the claim follows.\qed
\end{proof}

\section{Full Concrete Instantiation}\label{sec:supplement-instantiation}

In order to efficiently instantiate our generic construction, we require signature schemes and zero-knowledge proofs of knowledge which are compatible and support for flexible types of predicates.
Specifically, our construction can thus be instantiated using the following primitives:
\begin{itemize}
    \item BBS multi-message signatures,
    \item Schnorr-based zero-knowledge proofs of knowledge using the Fiat-Shamir heuristic, and
    \item efficient OR proofs for Schnorr.
\end{itemize}

In the following we briefly recap these primitives and specify the assumptions under which they are secure.

\subsection{BBS Signatures}

We briefly recall the BBS multi-message signature scheme.
It was first introduced in \cite{DBLP:conf/crypto/BonehBS04}, with the security proof used in our construction presented by \cite{DBLP:conf/eurocrypt/TessaroZ23a}.

Let
$$
  \e : \GG_1 \times \GG_2 \to \GG_T
$$
be an efficiently computable, non-degenerate bilinear map, where
$\GG_1$, $\GG_2$, and $\GG_T$ are cyclic groups of prime order $p$.

The BBS signature scheme consists of the following algorithms.

\begin{description}
\item[Parameter Generation.]
$\ParGen(1^\secpar,L)$ generates bilinear groups $(\GG_1,\GG_2,\GG_T,p,\e)$
of prime order $p$, samples generators
$$
  g_1\rexec\GG_1,
  \qquad
  g_2\rexec\GG_2,
  \qquad
  h_1,\ldots,h_L\rexec\GG_1,
$$
and outputs
$$
  \pp   \gets \left( \GG_1,\GG_2,\GG_T,p,\e, g_1,g_2,h_1,\ldots,h_L \right).
$$
The message space is $\mathcal{M}=\Z_p$.

\item[Key Generation.]
$\KGen(\pp)$ samples $x\rexec\Z_p^*$ and computes
$$
    X_1 \gets g_1^x
    \qquad\text{and}\qquad
    X_2 \gets g_2^x.
$$
It sets $\isk\gets x$ and $\ipk\gets(X_1,X_2)$.

Note that publishing $X_1=g_1^x$ in addition to $X_2=g_2^x$ does not require a stronger assumption for the unforgeability of BBS.
Indeed, the $q$-SDH instance of \cref{def:q-sdh} already contains both $g_1^x$ and $g_2^x$, so the standard reduction can expose both components as part of the verification key.
The additional component $X_1$ is used only to enable simulation of the randomized presentation proof.
\item[Sign.]
To sign a message $\vec{m}=(m_1,\ldots,m_L)\in\Z_p^L$, the signer samples $e\rexec\Z_p\setminus\{-\isk\}$ and computes
$$
  B \gets g_1\prod_{i=1}^{L}h_i^{m_i} \quad\text{and}\quad  A \gets B^{\frac{1}{\isk+e}}.
$$
The signature is $\sig\gets(A,e)$.

\item[Verify.]
On input a verification key $w$, a message vector
$\vec{m}=(m_1,\ldots,m_L)$, and a signature $\sig=(A,e)$, the verifier first checks that $A\in\GG_1\setminus\{1_{\GG_1}\}$ and $e\in\Z_p$.
It then computes
$$
  B
  \gets
  g_1\prod_{i=1}^{L}h_i^{m_i}
$$
and outputs $1$ if and only if, for $\ipk=(X_1,X_2)$, it holds that
$$
      \e(X_1,g_2)=\e(g_1,X_2) \quad\text{and}\quad \e\left(A,X_2 \cdot g_2^e\right) = \e(B,g_2).
$$
\end{description}

We next define the complexity assumption necessary for stating the security of BBS signatures.
\begin{Definition}[$q$-strong Diffie-Hellman assumption ($q$-SDH)]\label{def:q-sdh}
Let $\GG_1$ and $\GG_2$ be cyclic groups of prime order $p$, let $\e:\GG_1\times\GG_2\rightarrow\GG_T$ be a bilinear map, and let $\langle g_1\rangle=\GG_1$ and $\langle g_2\rangle=\GG_2$.

For $q\in\mathbb{N}$, the \emph{$q$-strong Diffie-Hellman ($q$-SDH) problem} is defined as follows.
Given a tuple $\left(g_1, g_1^x, g_1^{x^2}, \allowbreak\ldots, g_1^{x^q}, g_2, g_2^x \right)$, where $x\rexec\Z_p$, compute a $\left(c,g_1^{\frac{1}{x+c}}\right)\in(\Z_p\setminus\{-x\})\times\GG_1$.

The \emph{$q$-SDH assumption} holds in $(\GG_1,\GG_2)$ if, for every \PPT adversary $\adv$, there exists a negligible function $\negl$ such that
\[
  \Pr\left[
    \begin{array}{c}
      c\in\Z_p,\quad x+c\neq 0,\quad
      A=g_1^{\frac{1}{x+c}}
    \end{array}
    :
    \begin{array}{c}
      x\rexec\Z_p;\\
      (c,A)\gets
      \adv\left(
        g_1,
        g_1^x,
        g_1^{x^2},
        \ldots,
        g_1^{x^q},
        g_2,
        g_2^x
      \right)
    \end{array}
  \right]
  \leq
  \negl(\secpar).
\]
\end{Definition}

Now, the following lemma regarding the security of BBS signatures holds:
\begin{Lemma}
The BBS multi-message signature scheme described above is EUF-CMA secure in the sense of \cref{def:euf-cma} under the $q$-strong Diffie-Hellman assumption in the underlying bilinear groups.
\end{Lemma}

The proof of this lemma can be found in \cite{DBLP:conf/eurocrypt/TessaroZ23a}.

\subsection{Schnorr-based Zero-knowledge Proofs for BBS Signatures}

We next recall the standard zero-knowledge proof of knowledge associated with BBS signatures.
The proof follows the structure of \cite{DBLP:journals/joc/Schnorr91} and its generalizations, cf. \cite{DBLP:phd/de/Bangerter2005}.
These interactive protocols are then made non-interactive using the Fiat-Shamir (FS) transform, cf. \cite{DBLP:conf/crypto/FiatS86}.

Such a proof allows a prover to demonstrate possession of a valid BBS signature on a message vector without revealing the signature or the signed messages.
Moreover, the proof can be combined with additional zero-knowledge arguments showing that the hidden messages satisfy a public policy.

Let $\vec{m}=(m_1,\ldots,m_L)\in\Z_p^L$ be a message vector, let $\sig=(A,e)$ be a valid BBS signature on $\vec{m}$ under verification key $\pk$, and let $\policy:\Z_p^L\rightarrow\{0,1\}$ be an efficiently decidable policy.

The corresponding proof relation is defined by
$$
  \mathcal{R}_{\mathsf{BBS},\policy}
  :=
  \left\{
    \left(
      (\pk,\policy),
      (\sig,\vec{m})
    \right)
    \;\middle|\;
    \Verify(\pk,\vec{m},\sig)=1
    \;\land\;
    \policy(\vec{m})=1
  \right\}.
$$

Recall also that a $\Sigma$-protocol is \emph{special honest-verifier zero knowledge} if, for every valid public statement and every prescribed challenge, an efficient simulator can generate an accepting transcript for that challenge without knowing a witness, with the same distribution as an honest execution conditioned on this challenge.
Standard Schnorr representation protocols satisfy this property.

The proof system consists of the following algorithms.

\begin{description}
\item[Parameter Generation.]
 This algorithm samples a hash function
  $$
    H:\{0,1\}^*\rightarrow\Z_p
  $$
  modeled as a random oracle, where $p$ is as in the BBS parameters.
  
\item[Prove.]
  This algorithm takes as input a public key $\pk$, a policy $\policy$, a valid BBS signature $\sig=(A,e)$ on $\vec{m}=(m_1,\ldots,m_L)$, and a presentation context $\ctx$.

  The prover first recomputes $B  \gets g_1\prod_{i=1}^{L}h_i^{m_i}$ and then samples $r_1,r_2\rexec\Z_p^*$.
  It computes:
$$
  \bar{A}
  \gets
  A^{r_1r_2},
  \qquad
  D
  \gets
  B^{r_2},
  \qquad
  \bar{B}
  \gets
  D^{r_1}\bar{A}^{-e},
  \qquad
  r_3
  \gets
  r_2^{-1}.
$$

The prover then computes the Fiat-Shamir proof

\begin{align*}
  \pi'
  \rexec
  \NIZK\bigg[&
    (r_1,r_3,e,m_1,\ldots,m_L):\\
    &\bar{B} = D^{r_1}\bar{A}^{-e} \ \land\ 
      g_1 = D^{r_3} \prod_{i=1}^{L}h_i^{-m_i}\ \land\ 
      \policy(m_1,\ldots,m_L)=1
  \bigg](\ctx).
\end{align*}

The first two relations are representation relations and can be proven using the standard Schnorr protocol.
The policy proof is composed using standard techniques, where we assume that the employed policy-proof techniques are special honest-verifier zero knowledge and special sound.

The resulting proof is $\pi  \gets (\bar{A},\bar{B},D,\pi')$.

\item[Verify.]
$\Verify_{\ZKP}(\pp_{\ZKP},\pk,\policy,\pi,\ctx)$ parses
$$
  \pi=(\bar{A},\bar{B},D,\pi').
$$
It first checks that the issuer public key is well formed, i.e.,
$$
    e(X_1,g_2)=e(g_1,X_2),
$$
and subsequently whether $\bar{A},\bar{B},D$ are valid non-identity group elements and that
$$
  \e(\bar{A},\pk)
  =
  \e(\bar{B},g_2).
$$

It then verifies $\pi'$ and outputs $1$ if both the pairing equation and the proof verify, and outputs $0$ otherwise.
\end{description}

\begin{definition}[Discrete-logarithm representation assumption]
Let $\G$ be a cyclic group of prime order $p$.
For $L\geq 1$, sample independent generators
$
    g_0,g_1,\ldots,g_L \rexec\G\setminus\{1\}.
$
The \emph{discrete-logarithm representation problem} is to compute a non-zero vector
$
    (\alpha_0,\alpha_1,\ldots,\alpha_L)
$
such that
$
    \prod_{i=0}^{L} g_i^{\alpha_i}=1.
$
We say that the \emph{discrete-logarithm representation assumption} holds in $\G$ if, for every \PPT adversary $\adv$, the probability of solving this problem is negligible in the security parameter.
\end{definition}

\begin{lemma}\label{lemma:bbs-special-soundness}
The Schnorr representation protocol used in the BBS presentation protocol is special sound.
Moreover, from an extracted witness
$(r_1,r_3,e,m_1,\ldots,m_L)$
satisfying
$$
    \bar B = D^{r_1}\bar A^{-e}
$$
and
$$
    g_1 = D^{r_3}\prod_{i=1}^L h_i^{-m_i},
$$
together with the verifier's pairing equation
$$
    e(\bar A,\ipk)=e(\bar B,g_2),
$$
one can efficiently reconstruct a valid BBS signature on
$\mathbf m=(m_1,\ldots,m_L)$,
except with negligible probability under the corresponding discrete-logarithm representation assumptions.
Furthermore, AND composition with other special-sound protocols preserves special soundness.
\end{lemma}
\begin{proof}
The proof follows a standard argument and is only made explicit for scientific rigor and self-containedness purposes.
Special soundness of the representation protocol follows from the standard Schnorr extraction argument.
Given two accepting transcripts with the same first message and distinct challenges, subtracting the corresponding response vectors and dividing by the challenge difference yields
$(r_1,r_3,e,m_1,\ldots,m_L)$
satisfying both representation equations.

Define
$$
    B := g_1\prod_{i=1}^L h_i^{m_i}.
$$
The second extracted relation gives
$$
    B=D^{r_3}.
$$
Assuming $r_1\neq0$, define
$$
    A:=\bar A^{r_3/r_1}.
$$
From the first extracted relation and the verifier's pairing equation, we obtain
$$
    e(\bar A,\ipk g_2^e)
    =
    e(D^{r_1},g_2).
$$
Raising both sides to $r_3/r_1$ yields
$$
    e(A,\ipk g_2^e)
    =
    e(D^{r_3},g_2)
    =
    e(B,g_2).
$$
Hence $(A,e)$ is a valid BBS signature on $\mathbf m$.

It remains to consider the degenerate cases.
If $r_1=0$, then the first relation together with the pairing check and $\bar A\neq1$ implies $X_2=g_2^{-e}$, which reveals the discrete logarithm of $X_2$ in $\G_2$.
If $r_3=0$, then the second relation gives
$$
    g_1\prod_{i=1}^L h_i^{m_i}=1,
$$
which yields a non-trivial representation relation among the independently sampled public generators in $\G_1$.
The first event therefore occurs only with negligible probability under the discrete-logarithm assumption in $\G_2$, while the second occurs only with negligible probability under the discrete-logarithm representation assumption in $\G_1$.

Finally, special soundness is preserved under AND composition:
two accepting transcripts with the same first message and distinct common challenges give two accepting transcripts with distinct challenges for every composed subprotocol, from which the corresponding witnesses can be extracted.\qed
\end{proof}

%\begin{Lemma}\label{lemma:enc-to-the-sky}
%Let $\ParGen$ additionally generate a public key $\mathsf{epk}$ of an IND-CPA-secure public-key encryption scheme.
%Assume that every witness is additionally encrypted under $\mathsf{epk}$, and that the proof relation additionally proves the respective consistency.
%If the resulting Fiat-Shamir proof is zero-knowledge and simulation sound in the programmable random-oracle model, then the augmented BBS proof of knowledge described above is zero-knowledge in the sense of \cref{def:zero-knowledge} and simulation-sound extractable in the sense of \cref{def:sse}.
%\end{Lemma}
%The lemma follows immediately from the standard Fiat-Shamir transformation and its weak simulation-extractability guarantees~\cite{DBLP:conf/indocrypt/FaustKMV12,DBLP:conf/crypto/FiatS86}, combined with the ``encryption to the sky'' paradigm as used, e.g., in~\cite{DBLP:conf/sacrypt/CamenischKLMNP15}, where the corresponding decryption key serves as the extraction trapdoor.

%This compiler is used to establish the stronger extraction property required by the generic security reduction.
%The practical instantiation evaluated in \cref{sec:evaluation} uses the standard BBS/Schnorr Fiat-Shamir proof system without this additional extraction layer.

\subsection{Efficient OR Composition of Zero-Knowledge Proofs}

We realize disjunctions using the standard OR composition of \cite{DBLP:conf/crypto/CramerDS94} to realize the required disjunction
$$
  \policy_\edge\left(     (\attr^{\edge_l}_j)_{j=1}^m   \right)=1   \ \lor\   \id^{\node_l}=\id^{\node_{l+1}}\,.
$$

For this, consider the general transformation: consider two $\Sigma$-protocols for relations $\mathcal R_0$ and $\mathcal R_1$, with challenges in $\Z_p$, where the prover knows a witness for at least one of the two statements.

To prove
$$
  x_0\in\mathcal L_{\mathcal R_0}  \ \lor\   x_1\in\mathcal L_{\mathcal R_1},
$$
the prover honestly executes the $\Sigma$-protocol for the branch for which it knows a witness and simulates an accepting transcript for the other branch.

We assume that the underlying $\Sigma$-protocols are special honest-verifier zero knowledge, so that an accepting transcript can be simulated for any prescribed branch challenge.
More precisely, assuming that the prover knows a witness for branch $b$, it chooses a random challenge $c_{1-b}\rexec\Z_p$ and response $z_{1-b}$ for the other branch and uses the corresponding simulator to compute a commitment $T_{1-b}$ yielding an accepting transcript $(T_{1-b},c_{1-b},z_{1-b})$.
For the real branch it computes an honest first message $T_b$.
The common Fiat-Shamir challenge is then computed as
$$
  c \gets H(x_0,x_1,T_0,T_1,\ctx),
$$
and the challenge for the real branch is defined as
$$
  c_b \gets c -  c_{1-b} \pmod p.
$$
The prover finally computes the honest response $z_b$ for challenge $c_b$. 
The resulting proof consists of
$$
  \pi_{\mathsf{OR}} =   (T_0,T_1,c_0,c_1,z_0,z_1),
$$
and the verifier checks
$$
  c_0 +  c_1=c\pmod p
$$
together with acceptance of both underlying $\Sigma$-protocol transcripts.

\medskip

For our security proof, we only require special soundness of the OR composition.
In particular, the final statement in the unforgeability experiment must be fresh with respect to all simulated presentation queries, so the additional quasi-unique-response condition used to handle malleability of proofs for previously simulated statements is not needed.
The standard OR composition preserves special soundness as follows.

\begin{lemma}\label{lemma:or-special-soundness}
Assume that the underlying $\Sigma$-protocols for relations $\mathcal{R}_0$ and $\mathcal{R}_1$ are special sound.
Then their standard OR composition is special sound.
\end{lemma}

\begin{proof}
The proof is again a standard argument.
Consider two accepting OR transcripts with the same first message $(T_0,T_1)$ and distinct overall challenges $c\neq c'$.
Let
$$
    c=c_0+c_1 \pmod p
    \quad\text{and}\quad
    c'=c'_0+c'_1 \pmod p
$$
be the respective challenge decompositions.
Since $c\neq c'$, there exists $b\in\{0,1\}$ such that $c_b\neq c'_b$.
Hence,
$
    (T_b,c_b,z_b) \text{ and } (T_b,c'_b,z'_b)
$
are two accepting transcripts of the underlying $\Sigma$-protocol for branch $b$ with the same first message and distinct challenges.
By special soundness of that protocol, one can efficiently extract a witness $w_b$ for $\mathcal{R}_b$.
This witness is also a valid witness for the disjunctive relation $\mathcal{R}_0\lor\mathcal{R}_1$.\qed
\end{proof}

The mutual exclusivity of the two branches for valid certified edges is not required for special soundness of the OR protocol.
It instead ensures that, within the complete path relation, an authenticated edge is interpreted either as a genuine policy-compliant edge or as an issuer-added padding loop.

\begin{corollary}\label{cor:complete-special-soundness}
Assume that the Schnorr representation protocols and the proof techniques used for the supported policies are special sound.
Then the complete interactive presentation protocol is special sound.
\end{corollary}

\begin{proof}
By \cref{lemma:bbs-special-soundness}, the BBS representation proofs and their AND compositions are special sound.
By \cref{lemma:or-special-soundness}, each genuine-edge-or-padding subprotocol is special sound as well.
Since the complete presentation protocol combines these subprotocols using a common challenge and consistently linked witness variables, special soundness of the complete protocol follows.\qed
\end{proof}

\begin{corollary}\label{cor:fs-wse}
Assume that the complete interactive presentation protocol is special sound and special honest-verifier zero knowledge, and that its challenge space is $\Z_p$.
Then its Fiat-Shamir transform is zero knowledge and fresh-statement weakly simulation extractable in the random-oracle model.
\end{corollary}

\begin{proof}
Let the issuer public key be $\ipk=(X_1,X_2)$.
The simulator samples $\rho,\delta\rexec\Z_p^*$ and sets
$$
    \bar A\gets g_1^\rho,
    \qquad
    \bar B\gets X_1^\rho,
    \qquad
    D\gets g_1^\delta.
$$
Hence, $\e(\bar A,X_2)=\e(\bar B,g_2).$
Moreover, $\bar A$ and $D$ are independently and uniformly distributed over
$\G_1\setminus\{1\}$, as in an honest presentation, while
$\bar B=\bar A^x$.
The remaining interactive transcript can then be generated using the usual special honest-verifier zero-knowledge simulator.
Programming the random oracle at the corresponding Fiat-Shamir query yields a simulated non-interactive proof with the same distribution as an honest proof.

For extraction, consider an adversary that has adaptive access to simulated proofs and outputs an accepting proof $\pi^*$ for a fresh statement $x^*$, i.e., one that was not previously submitted to the proof simulator.
Let $T^*$ be the first message of the corresponding transcript.
Since $x^*$ is fresh and the Fiat-Shamir input contains a domain-separated encoding of the complete public statement and all first messages, none of the simulated proofs has programmed the random oracle at the point corresponding to $(x^*,T^*)$.

Except with probability $1/p$, an adversary producing an accepting proof must have queried the random oracle at this point.
The extractor rewinds the adversary to this query, preserving its random tape and all preceding oracle and simulation transcripts, and answers the query with an independently sampled challenge $c'\neq c$.
By the standard forking argument, whenever the original accepting proof is produced with non-negligible probability, this yields with non-negligible probability a second accepting transcript with the same first message $T^*$ and a different challenge.

Special soundness of the complete interactive representation protocol then yields the corresponding representation witnesses.
Applying \cref{lemma:bbs-special-soundness} to every BBS component reconstructs valid BBS signatures except with negligible probability under the assumptions of that lemma, and therefore yields a witness $w^*$ satisfying $\mathcal{R}(x^*,w^*)=1$.
The resulting extraction probability satisfies the polynomial relation required by fresh-statement weak simulation extractability.
The argument is the fresh-statement specialization of the weak simulation-extractability extraction technique of Faust et al.~\cite{DBLP:conf/indocrypt/FaustKMV12}.\qed
\end{proof}
\section{Additional Extension Details}\label{sec:extension-details}
\subsection{Extension to Multi-Path Support}
Multi-path communication can further reduce the trust placed in individual repeater nodes.
On a single path, compromising any one trusted repeater is sufficient to compromise the transmitted key material.
A natural countermeasure is therefore to distribute the secret across multiple paths, for instance using secret sharing as considered for multi-path QKD, cf.~\cite{acns2026}.
With a $t$-out-of-$t$ sharing, an adversary must compromise at least one node on every path to recover the secret, while a $k$-out-of-$t$ sharing requires compromising nodes on at least $k$ different paths.
For this argument it is essential that the paths are disjoint except for their common source and destination.
More specifically, it is required that they are \emph{node-disjoint}, except for their common source and destination, rather than merely link-disjoint.
For example, in a star-shaped network, several paths may use different links while all traversing the same central repeater.
Such a repeater would then constitute a single point of failure and could obtain all shares despite link-disjointness.
Consequently, the relevant security requirement is that no intermediate node occurs on more than one of the selected paths.

Our construction can be extended to provide this guarantee without revealing the identities of the traversed nodes, following the same approach based on scope-exclusive pseudonyms as used in~\cite{acns2026}, cf. also \cite{DBLP:conf/eurocrypt/NaorPR99} and \cite{DBLP:books/daglib/p/BichselCDEKLNP15}.
For this purpose, each node $\node$ is additionally certified with a high-entropy secret value $x_\node$ as a hidden attribute. The issuer assigns one stable secret to each node identifier throughout the same graph epoch because signing a different secret for an already certified node would break the disjointness test.
For a common context $\ctx$ (which in the following is assumed to be unique per presentation), every intermediate node $\node_k$ appearing in a presented path derives the session-specific pseudonym
$$
    \mathsf{nym}_{\node_k,\ctx}=H(\ctx)^{x_{\node_k}}\,,
$$
where $H$ hashes into an appropriate group.
The presentation proof is extended to show in zero knowledge that $\mathsf{nym}_{\node_k,\ctx}$ was correctly derived from the same secret $x_{\node_k}$ contained in the certified attributes of $\node_k$.
Since equal secrets yield equal pseudonyms for the same $\ctx$, while pseudonyms derived for different session identifiers remain unlinkable, the verifier can compare the pseudonym sets obtained for all paths associated with the same $\ctx$ and reject whenever two paths share an intermediate node.
Thus, pairwise disjointness of these pseudonym sets establishes node-disjointness without revealing the identities of the traversed nodes.

The direct construction above applies only to unpadded paths, but incurs only a small additional cost over the underlying presentation proof.
Supporting padded paths requires distinguishing genuine node occurrences from repetitions introduced solely for padding.
This can be achieved by applying a similar technique as for the edge-policy check in \cref{s:generic_construction}.
For every position $k>1$, the prover proves that either the current node is a repetition of the previous node, or that it is a genuine path node satisfying the node policy and that its pseudonym was correctly derived from the certified secret attribute.
Concretely, the corresponding part of the proof relation can be expressed as
\begin{align*}
    \id^{\node_k}=\id^{\node_{k-1}}
    \ \lor\
    \bigg(
        &\Verify_{\MMS_\node}\left(\ipk_\node,\left(\id_\GG,t,\id^{\node_k},\left(\attr^{\node_k}_i\right)_{i=1}^n\right),\sig_{\node_k}\right)=1
        \land\\
        &\policy_\node\left(\left(\attr^{\node_k}_i\right)_{i=1}^n\right)=1
        \land \mathsf{nym}_{k,\ctx}=H(\ctx)^{x_k}
    \bigg).
\end{align*}
Here, $x_k$ denotes the high-entropy secret value included among the certified attributes of $\node_k$.
If the first branch holds, the occurrence corresponds to padding and the prover may output an independently sampled random group element in place of a pseudonym.
Consequently, repeated padding occurrences of the same node do not create equal pseudonyms and therefore do not falsely violate the node-disjointness check.

This extension comes at a substantially higher proof cost than the direct construction for unpadded paths.
In particular, the node-signature verification, the node-policy predicates, and the pseudonym-correctness relation all appear inside one branch of an OR proof for every path position.
With a straightforward OR composition, the prover must therefore additionally simulate the alternative branch and provide responses for both branches, such that the cost of the node-related part of the proof can approach twice that of the unpadded construction.
The exact overhead depends on the concrete proof system and on whether common subrelations can be factored out or otherwise shared between the two branches.

\subsection{Adding nodes and edges}
  Networks may evolve over time as nodes or edges are added, removed, or replaced.
  If certified nodes or edges are removed, the graph has to be certified again, similarly to existing graph-signature approaches~\cite{DBLP:conf/ccs/Gross14}.
  However, unlike these approaches, our construction naturally supports monotonic graph growth without requiring a new certification of the entire graph.
  New nodes or edges can simply be added by issuing the corresponding signatures under the existing graph identifier and binding tag, using the appropriate node or edge identifiers and attributes.  

\subsection{Proving Non-Connectivity}
  Although not required for our main application, the construction can also be extended to prove that two nodes are not connected, thereby achieving another important feature presented by \cite{DBLP:conf/ccs/Gross14}.
  For this purpose, the issuer additionally assigns each node an identifier of its connected component in the underlying undirected network and includes this identifier among the certified node attributes.
  
  To prove that two nodes are disconnected, the prover shows possession of valid signatures on both nodes and proves in zero knowledge that their certified component identifiers are different.
  This only requires an inequality proof over two certified attributes and therefore adds constant-size overhead, independently of the size of the graph.

  We note that this extension is inherently not fully compatible with the efficient monotonic graph growth discussed above.
  Adding nodes or edges within an existing connected component leaves the identifiers unchanged.
  However, adding an edge between two previously disconnected components merges them, such that the affected nodes have to be re-certified with a common component identifier.

\section{On the Use of Generative AI}
Generative AI tools, in particular ChatGPT models from the GPT-5 and GPT-6 families, were used during the preparation and revision of this manuscript for language editing, consistency checks, and as an interactive aid for discussing and stress-testing selected technical arguments.
This included identifying possible weaknesses in proof formulations and suggesting alternative ways of structuring or presenting individual arguments.
Any such suggestions were subsequently assessed, rederived where necessary, and independently verified by the authors before being incorporated into the manuscript.
The tools were not treated as authoritative sources, and references or technical claims suggested by them were not used without independent verification.
The research problem, overall construction, and final technical results remain the responsibility of the authors, who take full responsibility for the content of the paper.

\end{document}